\documentclass[journal]{IEEEtran}
\usepackage{amsmath,amssymb,amsthm,bm,mathtools,mathrsfs}
\usepackage{graphicx}
\usepackage{enumitem}
\usepackage{mathrsfs}
\usepackage{algorithm}
\usepackage{algpseudocode} 
\usepackage{stfloats}
\usepackage{graphicx}
\usepackage{cite}
\usepackage{verbatim}
\usepackage{enumitem}
\usepackage[caption=false,font=footnotesize]{subfig}
\DeclareMathOperator{\tr}{tr}

\newtheorem{lemma}{Lemma}

\newtheorem{proposition}{Proposition}
\newtheorem{theorem}{Theorem}
\newtheorem{remark}{Remark}

\graphicspath{{code/Fig/}}

\begin{document}
	\title{Antenna Placement for Monostatic Near-Field Wireless Sensing}
	\author{
		Jinjian Liu,~\IEEEmembership{Graduate Student Member,~IEEE}, 
		Xianxin~Song,~\IEEEmembership{Member,~IEEE}, 
		and Xianghao~Yu,~\IEEEmembership{Senior Member,~IEEE}
		\thanks{This paper has been presented in part at the IEEE International Mediterranean Conference on Communications and Networking (MeditCom), Cagliari, Italy, July 2026 \cite{liumeditcom2026}.}
		\thanks{Jinjian Liu, Xianxin Song, and Xianghao Yu are with the Department of Electrical Engineering, City University of Hong Kong, Hong Kong, China (e-mail: jinjian.liu@my.cityu.edu.hk, xianxin.song@cityu.edu.hk, and alex.yu@cityu.edu.hk). Xianxin Song is the corresponding author.}
	}
	
	\maketitle
	\begin{abstract}	
		The emergence of movable antenna technology enables flexible antenna placement, allowing transceiver antennas to more effectively exploit spatial degrees of freedom. In this paper, we investigate a monostatic near-field sensing system, aiming to minimize the worst-case squared position error bound (SPEB) over the entire near-field region by jointly optimizing the transmit and receive antenna placements, along with the transmit power allocation. Toward this end, we first derive the closed-form expression of SPEB by calculating the Cram{\'e}r-Rao bound (CRB) for estimating the target's angle and distance. It is shown that the derived SPEB is governed by the power-weighted moments of the transmit antenna locations and the moments of the receive antenna locations, motivating the development of a moment-based optimization algorithm. Functional analysis of the SPEB's moment structure proves that the optimal antenna distributions can be realized by a centro-symmetric structure, and shows that the worst-case target is located at the array broadside on the Rayleigh boundary. Moreover, by leveraging moment-based analysis and the Richter-Tchakaloff theorem, we derive closed-form antenna distributions with three points supported at the aperture center and two edges. Specifically, the optimal transmit antenna distribution can be exactly realized by activating only three transmit antennas while the optimal receive antenna distribution consists of three clusters. Numerical results show that the proposed closed-form design significantly outperforms conventional antenna placements and power allocation benchmark schemes.
	\end{abstract}
	\section{Introduction}
	
	The forthcoming sixth generation (6G) of wireless networks is envisioned to support emerging applications such as unmanned aerial vehicles (UAVs), vehicle-to-everything (V2X), and smart cities \cite{ChenJCIN2025, JornetProIEEE2025}. These applications require not only reliable connectivity but also timely and accurate perception of the surrounding physical environment. In particular, precise target localization and motion tracking are essential for autonomous UAV navigation, collision avoidance, and cooperative driving in V2X systems, as well as traffic monitoring and infrastructure management in smart cities. By extracting location, motion, and environmental information from echo signals, wireless sensing can provide timely and wide-area situational awareness for these emerging applications. Consequently, wireless sensing is expected to become a key service of 6G wireless networks, imposing stringent requirements on sensing performance \cite{DuCST2025}.
	
	Wireless sensing systems commonly rely on multiple-input multiple-output (MIMO) architectures to exploit the spatial phase diversity provided by multiple antennas \cite{LiTSP2008, BekkermanTSP}. As MIMO sensing systems evolve toward higher carrier frequencies and larger array apertures, the near-field region expands to cover practically relevant sensing distances, rendering near-field effects non-negligible \cite{MaICST2026, liuOJCS2023}. In the near-field region, the conventional far-field plane-wave assumption becomes inaccurate because it
	neglects the wavefront curvature and the resulting range-dependent
	phase differences across the array. 
	To accurately characterize near-field propagation, spherical wavefronts are required, in which the spatial phase variations across the array depend jointly on the target's angle and range.
	These angle- and range-dependent responses provide more geometric information for target estimation, enabling joint angle-range estimation and thereby facilitating target localization. However, conventional MIMO sensing systems typically employ uniform antenna architectures, such as uniform linear arrays (ULAs) and uniform planar arrays (UPAs).
	Such uniform antenna configurations may create redundant spatial signatures across antennas, leaving the additional spatial degrees of freedom (DoFs) offered by flexible antenna placement largely unexploited \cite{MaTWC0824}.
	
	To address the above limitations, movable antenna (MA) technology has recently attracted significant attention in both wireless communication and sensing systems \cite{ZhuCST, ZhuCM}. By enabling flexible adjustment of antenna positions within a given spatial region, MA system provides the possibility of non-uniform antenna placement \cite{ZhuCST}.  This flexibility enables the sensing system to fully exploit the spatial DoFs arising from variations in the received signals across different antennas, thereby enhancing sensing performance. Similarly, related reconfigurable-antenna architectures, such as fluid antennas (FA) \cite{NewCST2025} and pinching antennas \cite{ZhengWC2026}, have also been investigated to achieve enhanced spatial reconfigurability.
	
	Recently, the antenna placement design has been widely exploited to enhance the performance of wireless sensing and integrated sensing and communications (ISAC) in the near-field region \cite{GazzahTAP2014, wang202512arxiv, DingTWC2026, SunIoTJ2025, ZhouGlobalcom2024, YangTCCN1025}. 
	Specifically, the authors in \cite{GazzahTAP2014} considered a single source sensing scenario using
	a linear receive array. They derived the Cram{\'e}r-Rao bound (CRB) expressions for estimating angle and range with respect to (w.r.t.) antenna positions, and numerically demonstrated that centro-symmetric array geometries can reduce both range and angle estimation errors compared with traditional ULA. 
	In \cite{wang202512arxiv}, the authors investigated single-source sensing using a non-uniform array and optimized the antenna positions to minimize the worst-case CRB for angle and range estimation by utilizing an iterative algorithm. In \cite{DingTWC2026}, the authors considered a full-duplex near-field ISAC system with multiple transmit and receive antennas, where the transmit beamformers, sensing-signal covariance matrices, receive beamformers, transceiver antenna positions, and uplink power allocation were jointly optimized to maximize a weighted sum of sensing and communication rates. In \cite{SunIoTJ2025}, the authors introduced rotatable MAs into near-field ISAC to exploit both position and orientation DoFs, in which the transmit beamformers, antenna positions, and rotations were jointly designed under angle-range CRBs and communication-rate metrics. 
	In \cite{ZhouGlobalcom2024}, the authors studied an FA-assisted near-field ISAC system enabled by a reflecting surface, where the communication beamformers, sensing-signal covariance matrix, reflecting coefficients, and FA position vector were jointly optimized to minimize the sum CRB for estimating the target's range and angle, subject to communication and power constraints by utilizing a double-loop iterative algorithm.
	In \cite{YangTCCN1025}, an FA-enabled near-field integrated sensing, computing, and semantic communication (NF-ISCSC) framework was investigated, where the communication beamformers, transmit covariance matrix, FA positions, and semantic extraction ratios were jointly optimized to maximize the worst-case semantic secrecy rate under sensing accuracy, computation, latency, and power constraints via an alternating optimization (AO) approach. 
	
	Nevertheless, prior studies on antenna placement design for near-field wireless sensing \cite{GazzahTAP2014, wang202512arxiv} and ISAC \cite{DingTWC2026, SunIoTJ2025, YangTCCN1025, ZhouGlobalcom2024} have left several important issues unresolved.
	First, existing studies on antenna placement for near-field sensing \cite{wang202512arxiv} and ISAC \cite{DingTWC2026, SunIoTJ2025, YangTCCN1025, ZhouGlobalcom2024} mainly rely on iterative algorithms for antenna position optimization, lacking optimal array geometry analysis and normally incur high computational complexity, especially for large-scale arrays. 
	Second, existing research on antenna placement for near-field wireless sensing and ISAC mainly focuses on single-sided array geometry optimization, considering either transmit-side array geometry designs \cite{YangTCCN1025} or receive-only array geometry designs \cite{GazzahTAP2014, wang202512arxiv}. Joint transmit/receive design is more challenging because the two array geometries jointly affect the angle
	and range information, while the transmit geometry is
	also coupled with the power allocation. 
	Although joint transmit/receive
	antenna position optimization has been investigated in \cite{DingTWC2026, SunIoTJ2025}, these works considered transmit and receive arrays deployed in bistatic ISAC systems and do not directly analyze target-position estimation errors. Specifically, \cite{DingTWC2026} adopted sensing mutual information as the
	sensing metric, which does not characterize target estimation accuracy, whereas \cite{SunIoTJ2025} minimized the sum of the CRBs for angle and range estimation, resulting in a physically inconsistent metric because the two quantities have different units.
	Consequently, the joint design of the antenna placements and the transmit power allocation for directly enhancing target's location estimation accuracy in monostatic near-field sensing system remains unexplored.
	
	To address these issues, we consider a monostatic near-field wireless sensing system consisting of a linear transmit array, a linear receive array, and a single target. The transmit antennas emit orthogonal waveforms over multiple snapshots. Based on this model, we investigate the joint design of transmit and receive antenna placements as well as the transmit power allocation, aiming to improve the worst-case target position estimation performance, i.e., to minimize the worst-case squared position error bound (SPEB).
	The main contributions of the paper are summarized as follows:
	\begin{itemize}
		\item First, we derive closed-form CRB expressions for joint angle and range estimation in a monostatic near-field sensing system. Considering the inherent unit inconsistency between the CRBs for estimating angle and range, we introduce the SPEB as the sensing performance metric and derive its implicit expression.
		\item Next, we formulate a worst-case SPEB minimization problem by jointly optimizing the transmit/receive array placements and transmit power allocation. To facilitate the analysis, we temporarily relax the minimum inter-spacing constraints of adjacent transmit/receive antennas and reformulate a distribution-based optimization problem. To solve this problem, we first employ functional analysis to prove that the optimal distributions can be chosen to be centro-symmetric and accordingly obtain the worst-case target location. Then, by leveraging moment-based analysis and the Richter-Tchakaloff theorem, we show that the optimal distributions admit a three-point structure supported at the aperture center and its two edges.
		\item  Then, we prove that the optimal transmit-side distribution can be exactly realized by activating only three transmit antennas located at the center and two edges with a closed-form power allocation policy that assigns equal power to the edge antennas and the remaining power to the center antenna. This observation can significantly reduce the number of active radio frequency (RF) chains and the associated hardware cost in practical implementations. At the receiver, we develop an efficient three-cluster deployment strategy that approaches the optimal antenna distribution while satisfying the minimum inter-element spacing constraints.
		\item Finally, we present numerical results to validate the proposed closed-form design. The mean squared error (MSE) achieved by the maximum likelihood estimator (MLE) is shown to be identical to the derived SPEB in the high signal-to-noise ratio (SNR) regime, which validates the SPEB derivation. In addition, the proposed joint design consistently outperforms conventional array geometries and power allocation schemes under various SNRs, aperture sizes, and numbers of antennas with negligible additional computational cost, demonstrating the effectiveness of the proposed design strategy. 
	\end{itemize}
	\textit{Notations:}
	We use normal-face letters to denote scalars,
	boldface lowercase and uppercase letters to denote column vectors and matrices, respectively.
	The imaginary unit is denoted by $\jmath = \sqrt{-1}$.
	The real-part operator is denoted by $\Re\{\cdot\}$.
	The nearest-integer rounding operator is denoted by
	$\operatorname{round}(\cdot)$.
	For a vector \(\mathbf z\), \(z_i\) denotes its \(i\)-th entry, and
	\(\operatorname{diag}(\mathbf z)\) denotes the diagonal matrix with
	\(\mathbf z\) on its main diagonal. 
	For a matrix \(\mathbf Z\),
	\(\operatorname{tr}(\mathbf Z)\) denotes its trace, and $[\mathbf Z]_{i:j,k:\ell}$ denotes the
	submatrix formed by rows $i$ through $j$ and columns $k$ through $\ell$.
	The Kronecker product, transpose, conjugate transpose, probability, expectation, variance and covariance operations are denoted by $\otimes$, $(\cdot)^{\mathrm{T}}$, $(\cdot)^{\mathrm{H}}$, $\Pr(\cdot)$, $\mathbb{E}(\cdot)$, $\mathrm{Var}(\cdot)$ and $\mathrm{Cov}(\cdot,\cdot)$, respectively.
	The $N\times N$ identity matrix is denoted by $\mathbf I_N$.
	The set of $M \times N$ real and imaginary matrices is denoted by $\mathbb{R}^{M \times N}$ and $\mathbb{C}^{M \times N}$, respectively.
	The set of \(N\)-dimensional non-negative real vectors is denoted by \(\mathbb{R}_{+}^N\).
	The set of complex numbers is denoted by $\mathbb{C}$.
	The notation $\arg\max$ denotes the set of maximizers.
	A circularly symmetric complex Gaussian vector with mean $\bm \mu$ and covariance $\mathbf A$
	is denoted by $\mathcal{CN}(\bm \mu,\mathbf A)$.
	The space of real-valued $\vartheta$-integrable functions on a measure
	space $(\mathcal{X},\vartheta)$ is denoted by
	$L^1_{\mathbb{R}}(\mathcal{X},\vartheta)$.
	The set of $n\times n$ real symmetric positive definite matrices is denoted by $\mathbb{S}^{n}_{++}$,
	and $\preceq$/$\succeq$ denotes the Loewner partial order.
	\section{System Model}
	We consider a single-target monostatic near-field sensing scenario, as illustrated in Fig.~\ref{fig:nearfield_scenario}, which consists of a transmit array with $M$ antennas and a receive array with $N$ antennas. The two
	arrays are physically distinct but colocated on the same sensing
	platform \cite{LiTSP2008, WerfTSP2026}.
	In this work, we focus on a near-field target sensing scenario, motivated by the large aperture MIMO arrays and high-frequency transmitted signals \cite{wang202512arxiv}. 
	\begin{figure}[t]
		\centering
		\includegraphics[width=1\linewidth, height=1\linewidth, keepaspectratio]{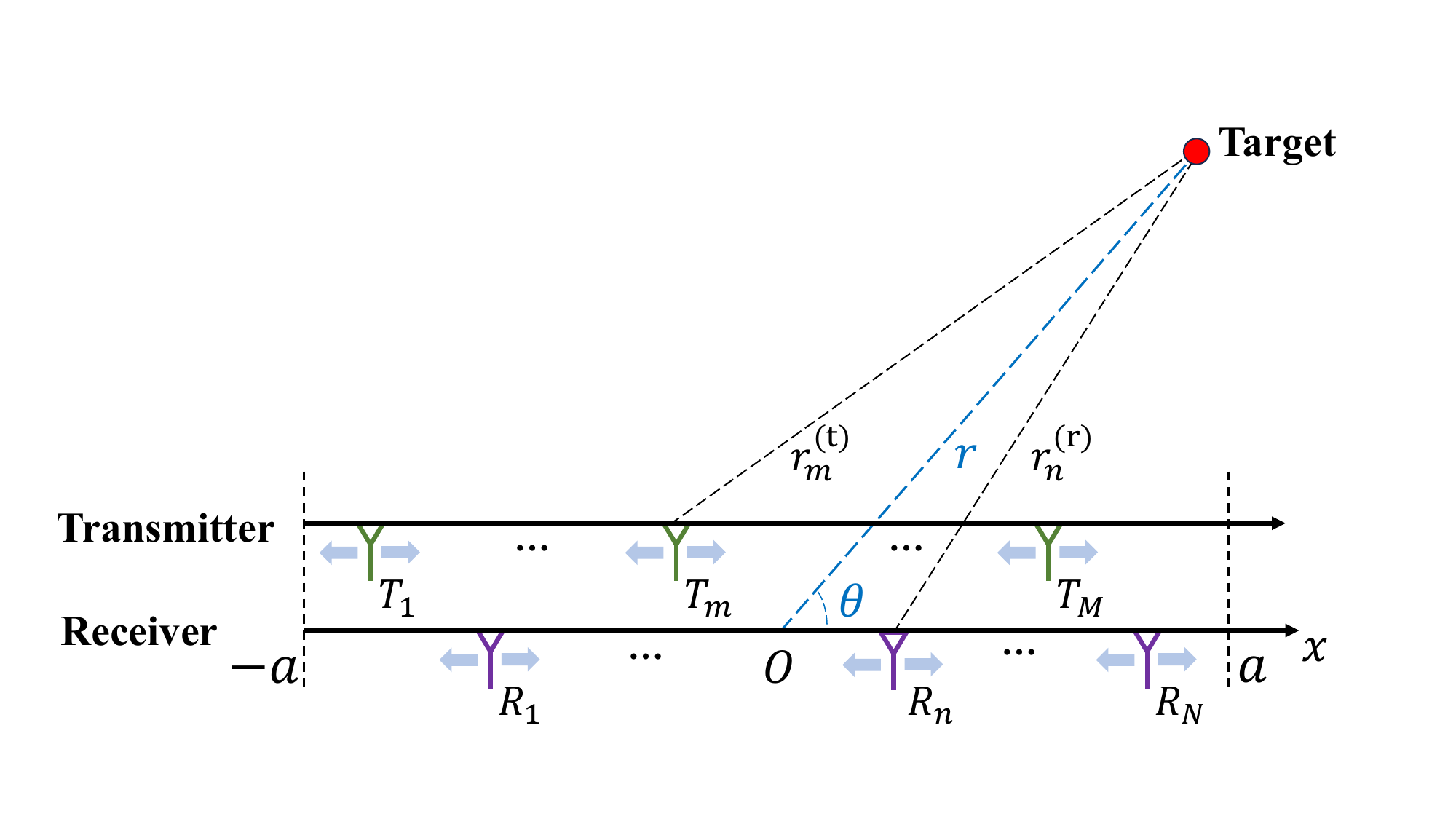}
		\label{fig:nearfield_scenario_mono}
		\vspace{-5mm}
		\caption{System model of the considered monostatic near-field sensing system.}
		\label{fig:nearfield_scenario}
	\end{figure}
	At snapshot $t \in \{1,\ldots,T\}$, the transmit array emits a deterministic probing signal vector $\mathbf{s}_t \in \mathbb{C}^{M\times 1}$. Over $T$ snapshots, the average sum transmit power across all transmit
	antennas is constrained as
	$	\frac{1}{T}\sum_{t=1}^{T}\|\mathbf{s}_t\|^2 \le P,	$
	where $P$ denotes the average sum-power budget of the transmit array.
	By stacking all $T$ snapshots, we obtain the transmit waveform matrix $\mathbf{S} = [\mathbf{s}_1,\mathbf{s}_2,\ldots,\mathbf{s}_T] \in \mathbb{C}^{M\times T}$.
	The waveform sequences assigned to different transmit antennas are assumed to be deterministic and mutually orthogonal across the $T$ snapshots \cite{NionTSP2010}, i.e.,
	\begin{equation}
		\frac{1}{T}\sum_{t=1}^{T}s_{m,t}s_{\ell,t}^{*}
		=
		\begin{cases}
			P\rho_m, & m=\ell,\\
			0, & m\neq \ell,
		\end{cases}
		\label{eq:waveform_orthogonality}
	\end{equation}
	where $s_{m,t}$ denotes the $(m,t)$-th entry of $\mathbf{S}$, and $\rho_m \ge  0$ is the fraction of transmit power allocated to the $m$-th transmit antenna, with $\sum_{m=1}^M \rho_m=1$. Consequently, the transmit covariance matrix is a diagonal matrix given by $
		\mathbf{R}_s
		= \frac{1}{T}\mathbf{S}\mathbf{S}^{\mathrm H}
		= P\,\operatorname{diag}(\bm \rho),
	$
	where $\bm{\rho}=[\rho_1,\ldots,\rho_M]^{\mathrm T} \in \mathbb{R}_{+}^M$ is the power-allocation vector. 
	\begin{remark}
		While orthogonal transmit waveforms are commonly adopted in wireless sensing
		systems \cite{BekkermanTSP, GodrichTSP0711}, the CRB derivation framework in this paper can also be extended to
		non-orthogonal waveforms. However, the resulting cross-correlation among non-orthogonal waveforms would yield a highly complicated CRB expression, making it challenging to obtain a closed-form solution for the array placement problem.
	\end{remark}
	
	We assume that the transmit and receive antennas are positioned along the $x$-axis within the aperture of length $D=2a$. Let $T_m =(x_m^{(\mathrm{t})},0) $ and $R_n = (x^{(\mathrm{r})}_n, 0)$ denote the locations of the $m$-th transmit antenna and the $n$-th receive antenna, respectively, where $x^{(\mathrm{t})}_m\in[-a,a]$ for $m=1,\ldots,M$ and $x^{(\mathrm{r})}_n\in[-a,a]$ for $n=1,\ldots,N$. The antenna locations are collected into vectors $\mathbf{x}_{\mathrm t}=[x^{(\mathrm{t})}_1,\ldots,x^{(\mathrm{t})}_M]^\mathrm{T}$ and $\mathbf{x}_{\mathrm r}=[x^{(\mathrm{r})}_1,\ldots,x^{(\mathrm{r})}_N]^\mathrm{T}$. 
	Without loss of generality, the antennas are indexed in ascending order of their locations, i.e., $-a \le x^{(\mathrm{t})}_1 < x^{(\mathrm{t})}_2 \dots < x^{(\mathrm{t})}_M \le a$ and $-a \le x^{(\mathrm{r})}_1 < x^{(\mathrm{r})}_2 \dots < x^{(\mathrm{r})}_N \le a$. 
	To mitigate the mutual coupling between adjacent MAs, a minimum inter-element spacing $d$ is imposed between any two adjacent MAs.
	The position of the target can be characterized either by its Cartesian coordinates $\mathbf{p}=[p_1,p_2]^{\mathrm{T}}$, or equivalently, by the coordinates $\bm{\eta}=[u,r]^{\mathrm{T}}$. Here, $u=\cos\theta$ and $r = \sqrt{p_1^2+p_2^2}$ represents the range from the array origin to the target, with $\theta\in[0,\pi]$ being the angle between the target-to-array-center line and the positive $x$-axis.
	
	In this work, we consider the target to be located within the radiating near-field region, denoted by $\Gamma$, whose lower and upper
	boundaries are characterized by the Fresnel distance $d_{\mathrm{F}}$ and the Rayleigh distance $d_{\mathrm{R}}$, respectively. For a given angle $\theta$, the Fresnel and Rayleigh distances are formulated as $d_{\mathrm{F}} = 0.62\sqrt{\frac{D^3\sin^3\theta}{\lambda}}$ and $d_{\mathrm{R}} = \frac{2D^2\sin^2\theta}{\lambda}$, respectively \cite{Balanis, liuOJCS2023,LuTWC2022}.
	Within $\Gamma$, the uniform spherical wave (USW) channel model is valid, which assumes uniform path loss amplitudes across the array elements while accounting for the phase variations induced by spherical wavefronts \cite{liuOJCS2023}. 
	
	Let $r_m^{(\mathrm{t})}(\bm{\eta})$ and $r_n^{(\mathrm{r})}(\bm{\eta})$ denote the distances from the target to the $m$-th transmit antenna and the $n$-th receive antenna, respectively, i.e.,
	$r_m^{(\mathrm{t})}(\bm{\eta}) = \sqrt{r^2 - 2rux_m^{(\mathrm{t})} + \bigl(x_m^{(\mathrm{t})}\bigr)^2}$ and $
			r_n^{(\mathrm{r})}(\bm{\eta}) = \sqrt{r^2 - 2rux_n^{(\mathrm{r})} + \bigl(x_n^{(\mathrm{r})}\bigr)^2}$.
	We assume that the sensing channel is dominated by a line-of-sight (LoS) link \cite{Song2026arxiv}. In particular, let $\mathbf{a}_{\mathrm t}(\mathbf{x}_{\mathrm t},\bm{\eta})$ and $\mathbf{b}_{\mathrm r}(\mathbf{x}_{\mathrm r},\bm{\eta})$ denote the transmit and receive steering vectors respectively, i.e.,
	\begin{equation}
		\begin{aligned}
			\mathbf{a}_{\mathrm t}(\mathbf{x}_{\mathrm t},\bm{\eta})
			&=\big[e^{-\jmath \frac{2\pi}{\lambda} r^{(\mathrm{t})}_1(\bm{\eta})},\,
			e^{-\jmath \frac{2\pi}{\lambda} r^{(\mathrm{t})}_2(\bm{\eta})},\,\ldots,\,
			e^{-\jmath \frac{2\pi}{\lambda} r^{(\mathrm{t})}_M(\bm{\eta})}\big]^{\mathrm{T}}, \\
			\mathbf{b}_{\mathrm r}(\mathbf{x}_{\mathrm r},\bm{\eta})
			&=\big[e^{-\jmath \frac{2\pi}{\lambda} r^{(\mathrm{r})}_1(\bm{\eta})},\,
			e^{-\jmath \frac{2\pi}{\lambda} r^{(\mathrm{r})}_2(\bm{\eta})},\,\ldots,\,
			e^{-\jmath \frac{2\pi}{\lambda} r^{(\mathrm{r})}_N(\bm{\eta})}\big]^{\mathrm{T}},
			\label{eq:steering}
		\end{aligned}
	\end{equation}
	where $\lambda$ denotes the wavelength of the transmitted signal. Using the Fresnel approximation in the near-field \cite{liuOJCS2023} of $r_m^{(\mathrm{t})}(\bm{\eta})$ and $r_n^{(\mathrm{r})}$, after removing the common range-dependent phase term, the transmit and receive steering vectors can be written as
	\begin{equation}
		\begin{aligned} \label{eq:steering_fresnel}
			\mathbf a_{\mathrm t}(\mathbf x_{\mathrm t},\boldsymbol{\eta})
			&\approx
			\Big[
			e^{\jmath\frac{2\pi}{\lambda}\psi(x^{(\mathrm{t})}_1;\bm{\eta})},
			\ldots,
			e^{\jmath\frac{2\pi}{\lambda}\psi(x^{(\mathrm{t})}_M;\bm{\eta})}
			\Big]^{\mathrm T},
			\\
			\mathbf b_{\mathrm r}(\mathbf x_{\mathrm r},\boldsymbol{\eta})
			&\approx
			\Big[
			e^{\jmath\frac{2\pi}{\lambda}\psi(x^{(\mathrm{r})}_1;\bm{\eta})},
			\ldots,
			e^{\jmath\frac{2\pi}{\lambda}\psi(x^{(\mathrm{r})}_N;\bm{\eta})}
			\Big]^{\mathrm T}.
		\end{aligned}
	\end{equation}
	where $	\psi(x;\bm{\eta})=ux-\frac{1-u^2}{2r}x^2$.
	The received signal at the receive array collected over $T$ snapshots can be expressed as \cite{LiTSP2008}:
	\begin{equation}
		\mathbf{y}_t = \alpha\, \beta(r)\,\mathbf{b}_{\mathrm r}(\mathbf x_{\mathrm r}, \bm{\eta}) \mathbf{a}_{\mathrm t}^{\mathrm{T}}(\mathbf x_{\mathrm t}, \bm{\eta}) \mathbf{s}_t + \mathbf{n}_t, \quad t=1,\ldots,T,
	\end{equation}
	where $\alpha\in\mathbb C$ denotes the unknown complex target-reflection coefficient associated with the target's radar cross-section (RCS), $\beta(r)$ denotes the range-dependent round-trip channel gain that is positive and non-increasing with $r$ \cite{liuOJCS2023}, and $\mathbf{n}_t\sim\mathcal{CN}(\bm{0},\sigma^2\mathbf{I}_N)$ denotes the additive white Gaussian noise (AWGN). Then, the received data matrix over $T$ snapshots can be written as
	\begin{equation}
		\mathbf Y
		=
		\alpha\, \beta(r)\,
		\mathbf b_{\mathrm r}(\mathbf{x}_{\mathrm r},\boldsymbol{\eta})\,
		\mathbf a_{\mathrm t}^{\mathrm T}(\mathbf{x}_{\mathrm t},\boldsymbol{\eta})\,
		\mathbf S
		+
		\mathbf N_0,
		\label{eq:raw_obs_matrix}
	\end{equation}
	where $\mathbf N_0	=[\mathbf n_1, \ldots, \mathbf n_T] \in \mathbb C^{N\times T}$.
	Vectorizing \eqref{eq:raw_obs_matrix} yields the vectorized observation model $\tilde{\mathbf y} \in \mathbb{C}^{NT \times 1}$ for subsequent parameter estimation:
	\begin{equation}
		\tilde{\mathbf y}
		=
		\operatorname{vec}(\mathbf Y)
		=
		\alpha\, \beta(r)\,
		\mathbf g
		+
		\tilde{\mathbf n},
		\label{eq:vectorized_statistic}
	\end{equation}
	where
	$
	\tilde{\mathbf n}
	\sim\mathcal{CN}(\mathbf 0,\sigma^2\mathbf I_{NT})
	$, 
	and
	$
			\mathbf g
			=
			\Big(
			\mathbf S^{\mathrm T}\mathbf a_{\mathrm t}(\mathbf{x}_{\mathrm t},\boldsymbol{\eta})
			\Big)
			\otimes
			\mathbf b_{\mathrm r}(\mathbf{x}_{\mathrm r},\boldsymbol{\eta}).
	$
	In \eqref{eq:vectorized_statistic}, both the target position
	$\bm{\eta}$ and the target-reflection
	coefficient $\alpha$ are unknown, with the former being our estimation objective. To evaluate the fundamental limits of the target position estimation performance, let $\mathrm{CRB}_{\bm{\eta}} \in \mathbb R^{2 \times 2}$ denote the CRB matrix for estimating $\bm{\eta}$, which will be derived in Section~\ref{sec: monostatic}.
	
	\section{Performance-Bound Derivation and Problem Formulation}
	\label{sec: monostatic}
	In this section, we first derive the CRB for target angle and range estimation to characterize the corresponding SPEB, and formulate the worst-case SPEB minimization problem for joint transmit/receive antenna placements and transmit power allocation.
	
	\subsection{Cramér-Rao Bound Derivation} \label{sec: monostatic-CRB}
	Based on the model in \eqref{eq:vectorized_statistic}, 
	define
	$\mathbf{w}(\mathbf z) = \alpha\, \beta(r)\,
	\mathbf g$ and $\mathbf z = [\bm{\eta}^{\mathrm T}, \bm{\alpha}^{\mathrm T}]^{\mathrm{T}}$ as the vector of unknown parameters, where $\bm{\alpha} = [\alpha_r,  \alpha_i]^{\mathrm{T}}$, with $\alpha_r$ and $\alpha_i$ denoting the real and imaginary parts of $\alpha$, respectively. 
	According to the derivation in \cite{kay1993fssp_estimation}, the Fisher information matrix (FIM) w.r.t. $\mathbf z$ is given by 
	\begin{equation}
		\begin{aligned}
			\mathbf F
			&=
			\frac{2}{\sigma^2}\,
			\Re\!\left\{
			\left(\frac{\partial \mathbf w}{\partial \mathbf z}\right)^{\mathrm H}
			\left(\frac{\partial \mathbf w}{\partial \mathbf z}\right)
			\right\} \in \mathbb R^{4 \times 4}.
			\label{eq:FIM_block_v}
		\end{aligned}
	\end{equation}
	Then, the CRB matrix for estimating $\bm\eta$ is given by
	the upper-left $2\times2$ principal submatrix of $\mathbf F^{-1}$ \cite{kay1993fssp_estimation, BoyerTSP}, i.e.,
	\begin{equation}
		\mathrm{CRB}_{\bm{\eta}}
		=
		\left[\mathbf F^{-1}\right]_{1:2,1:2}
		\in \mathbb{R}^{2\times2}.
		\label{eq:CRB_eta}
	\end{equation}
	Its diagonal entries provide lower bounds on the variances of any unbiased
	estimators of $u$ and $r$ \cite{BoyerTSP}, denoted by
	$\mathrm{CRB}_{u}$ and $\mathrm{CRB}_{r}$, respectively. To facilitate the CRB derivation, we define $m_k=\frac{1}{N}\sum_{n=1}^N (x_n^{(\mathrm{r})})^k
	$ and $s_k=\sum_{m=1}^M \rho_m (x_m^{(\mathrm{t})})^k$, respectively.
	\begin{lemma}
		The CRBs for estimating
		$u$ and $r$ are respectively given by
		\begin{equation}
			\begin{aligned}
				\mathrm{CRB}_{u}
				&=
				\kappa(r)
				\frac{M_{2}^{(\mathrm{r})}+M_{2}^{(\mathrm{t})}}
				{\left(M_{1}^{(\mathrm{r})}+M_{1}^{(\mathrm{t})}\right)
					\left(M_{2}^{(\mathrm{r})}+M_{2}^{(\mathrm{t})}\right)
					-\left(M_{3}^{(\mathrm{r})}+M_{3}^{(\mathrm{t})}\right)^{2}},
				\\
				\mathrm{CRB}_{r}
				&=
				\kappa(r)
				\frac{M_{1}^{(\mathrm{r})}+M_{1}^{(\mathrm{t})}}
				{\left(M_{1}^{(\mathrm{r})}+M_{1}^{(\mathrm{t})}\right)
					\left(M_{2}^{(\mathrm{r})}+M_{2}^{(\mathrm{t})}\right)
					-\left(M_{3}^{(\mathrm{r})}+M_{3}^{(\mathrm{t})}\right)^{2}}.
			\end{aligned}
			\label{eq.FIM-monostatic}
		\end{equation}
		where $\kappa(r)
		=\frac{\sigma^2\lambda^2}{8\pi^2TPN|\alpha|^2\beta^2(r)}$, and 
		\begin{align}
			M_1^{(\mathrm{r})}
			&=
			m_2+\frac{2u}{r}m_3+\frac{u^2}{r^2}m_4
			-
			\left(m_1+\frac{u}{r}m_2\right)^2,
			\label{eq:Mr1}
			\\
			M_2^{(\mathrm{r})}
			&=
			\frac{(1-u^2)^2}{4r^4}\left(m_4-m_2^2\right),
			\label{eq:Mr2}
			\\
			M_3^{(\mathrm{r})}
			&=
			\frac{1-u^2}{2r^2}
			\left(
			m_3+\frac{u}{r}m_4-m_1m_2-\frac{u}{r}m_2^2
			\right),
			\label{eq:Mr3}
			\\
			M_1^{(\mathrm{t})}
			&=
			s_2+\frac{2u}{r}s_3+\frac{u^2}{r^2}s_4
			-
			\left(s_1+\frac{u}{r}s_2\right)^2,
			\label{eq:Mt1}
			\\
			M_2^{(\mathrm{t})}
			&=
			\frac{(1-u^2)^2}{4r^4}\left(s_4-s_2^2\right),
			\label{eq:Mt2}
			\\
			M_3^{(\mathrm{t})}
			&=
			\frac{1-u^2}{2r^2}
			\left(
			s_3+\frac{u}{r}s_4-s_1s_2-\frac{u}{r}s_2^2
			\right).
			\label{eq:Mt3}
		\end{align}
	\end{lemma}
	\begin{IEEEproof}
		Please see Appendix~\ref{appendix:FIM_MN_monostatic}.
	\end{IEEEproof}
	
	\begin{remark}
		It is worth noting that the CRB expressions in \eqref{eq.FIM-monostatic} explicitly reveal their dependence on both the transmit and receive array geometries.
		However, $\mathrm{CRB}_{u}$ is dimensionless while $\mathrm{CRB}_{r}$ is measured in $\mathrm{m}^2$. It is inappropriate to use a weighted sum of CRBs for estimating angle and range as a sensing-performance metric. These observations motivate the adoption of the SPEB \cite{Wintit2010}, which measures the position estimation error directly.
	\end{remark}
	\vspace{-6mm}
	\subsection{Squared Position Error Bound Characterization}
    Mathematically, the SPEB is defined as the trace of the Cartesian CRB matrix and denoted by $\mathrm{SPEB}(\mathbf{x}_{\mathrm t},\mathbf{x}_{\mathrm r},\bm{\rho},\mathbf{p})$, i.e., $\mathrm{SPEB}(\mathbf x_{\mathrm t}, \mathbf x_{\mathrm r},\bm{\rho},\mathbf{p}) = \tr(\mathrm{CRB}_{\mathbf{p}})$ \cite{Wintit2010}, where $\mathrm{CRB}_{\mathbf{p}}\in \mathbb R^{2 \times 2}$ denotes the CRB matrix for
	estimating $\mathbf{p}$. Its relationship with the polar CRB is given in the following lemma.
	
	\begin{lemma}\label{lemma:CRB_Cartesian}
		The target's SPEB is given by
		\begin{equation}
			\begin{aligned}
				\label{eq. CRB-transform}
				\mathrm{SPEB}(\mathbf x_{\mathrm t}, \mathbf x_{\mathrm r},\bm{\rho},\mathbf{p}) =
				\frac{r^2}{1-u^2}\,\mathrm{CRB}_u
				+
				\mathrm{CRB}_r.
			\end{aligned}
		\end{equation}
	\end{lemma}
	
	\begin{IEEEproof}
		Please see Appendix~\ref{app.Jacobian-SPEB}.
	\end{IEEEproof}
	In this paper, we aim to ensure robust performance against the worst-case target location within the prescribed near-field region $\Gamma$. This is achieved by minimizing the maximum $\mathrm{SPEB}(\mathbf{x}_{\mathrm t},\mathbf{x}_{\mathrm r}, \bm{\rho}, \mathbf{p})$ over the near-field region, i.e.,
	\begin{subequations}
		\begin{alignat}{3}
			\text{(P1)}: & \quad & \operatorname*{minimize}_{\mathbf x_{\mathrm t}, \mathbf x_{\mathrm r}, \bm{\rho}\in \mathbb{R}_{+}^M}& \quad  \max_{\mathbf{p} \in \Gamma} \quad \mathrm{SPEB}(\mathbf x_{\mathrm t}, \mathbf x_{\mathrm r},\bm{\rho},\mathbf{p}) \\
			& & & \hspace{-2cm} \mathrm{s.t.} \quad x_m^{(\mathrm{t})} \in [-a, a], \quad \forall m=1, \dots, M, \\
			& & & \hspace{-1.5cm} \quad x_n^{(\mathrm{r})} \in [-a, a], \quad \forall n=1, \dots, N, \label{eq. cons-receive-range}\\
			& & & \hspace{-1.5cm}\quad x_{m+1}^{(\mathrm{t})} - x_m^{(\mathrm{t})} \ge d, \quad \forall m=1,\dots,M-1, \label{eq. cons-inter-spacing1}\\ 
			& & & \hspace{-1.5cm}\quad x_{n+1}^{(\mathrm{r})} - x_n^{(\mathrm{r})} \ge d, \quad \forall n=1,\dots,N-1, \label{eq. cons-inter-spacing2}\\ 
			& & & \hspace{-1.5cm}\quad \sum_{m=1}^{M}\rho_m = 1.
		\end{alignat}
	\end{subequations}
	where constraint \eqref{eq. cons-inter-spacing1} and \eqref{eq. cons-inter-spacing2} ensure minimum inter-element spacing to avoid mutual coupling. 
	
	It is worth noting that problem~(P1) is a highly coupled non-convex min-max problem involving
	the transmit antenna locations, receive antenna locations, and transmit
	power allocation. The minimum inter-element spacing constraints further
	couple adjacent antenna positions, making the optimization problem challenging. Although this non-convex optimization problem can be solved by conventional iterative methods, it could incur substantial complexity for
	large-scale arrays and provide limited structural insight. These challenges
	motivate the subsequent structural analysis and distributional reformulation.
	\section{Problem Reformulation and Closed-Form Design}
	In this section, we reformulate the problem (P1) as a distribution-based design problem. Then, we establish the optimality of centro-symmetric distributions, identify the worst-case target location, and subsequently derive the closed-form three-point structure via moment-based methods. Finally, we present a practical antenna placement and power-allocation realization strategy.
	\subsection{Problem Reformulation and Analysis}
	In problem (P1), the minimum spacing constraints \eqref{eq. cons-inter-spacing1} and \eqref{eq. cons-inter-spacing2} couple adjacent antenna positions, making the joint optimization mathematically intractable. To facilitate our analysis and gain theoretical insights, we temporarily relax constraints \eqref{eq. cons-inter-spacing1} and \eqref{eq. cons-inter-spacing2}. 
	Moreover, from \eqref{eq:Mr1}-\eqref{eq:Mr3}, we observe that $M_1^{(\mathrm{r})}$, $M_2^{(\mathrm{r})}$ and $M_3^{(\mathrm{r})}$ are closely related to the term $m_k = \frac{1}{N} \sum_{n=1}^N (x^{(\mathrm{r})}_n)^{k}$. This term can be naturally interpreted as the $k$-th sample moment of $N$ observations $\{(x^{(\mathrm{r})}_n)\}_{n=1}^N$, which are drawn from a certain distribution of a random variable $X_{\mathrm r}$. Similarly, $M_1^{(\mathrm{t})}$, $M_2^{(\mathrm{t})}$ and $M_3^{(\mathrm{t})}$ admit analogous
	representations in terms of the power-weighted transmit antenna
	location moments $s_k = \sum_{m=1}^M \rho_m(x^{(\mathrm{t})}_m)^{k}$, which can be viewed as the $k$-th moment of a random variable $X_{\mathrm t}$.
	Specifically, we define a sufficiently fine uniform candidate-location grid
	$\mathcal S=\{\zeta_1,\ldots,\zeta_L\}\subseteq[-a,a]$, where
	$|\mathcal S|=L$ and $L$ is a sufficiently large odd integer. Let $X_{\mathrm t}$ and $X_{\mathrm r}$ be discrete random variables supported on $\mathcal S$ that model the distribution of power-weighted transmit antennas locations and receiver antenna locations, respectively. The entries
	of $\mathbf{x}_{\mathrm r}$ can then be viewed as $N$ realizations of $X_{\mathrm r}$.
	As $N \to \infty$, functions $m_k$ can be identified with the moments of random variable $X_{\mathrm r}$, and similarly functions $s_k$ can be identified with the moments of random variable $X_{\mathrm t}$ as $M \to \infty$, i.e., $m_k = \mathbb{E}(X_{\mathrm r}^k)$ and $s_k = \mathbb{E}(X_{\mathrm t}^k)$. 
	Therefore, we can reformulate the discrete antenna placement and power allocation problem as a distribution design problem for $X_{\mathrm t}$ and $X_{\mathrm r}$. To characterize the distributions of $X_{\mathrm t}$ and $X_{\mathrm r}$, we introduce their probability mass functions $w_{\mathrm t}(\cdot)$ and $w_{\mathrm r}(\cdot)$ on $\mathcal S$ by $w(\zeta)= \Pr(X=\zeta)$ for all $\zeta\in\mathcal S$.

	This transformation leads to the reformulated problem $(\text{P1}')$ which is explicitly expressed as
	\begin{subequations}
		\begin{alignat}{2}
			(\text{P1}'): \quad & \operatorname*{minimize}_{w_{\mathrm r}(\cdot),w_{\mathrm t}(\cdot)} \quad && \max_{\mathbf{p} \in \Gamma} \quad \mathrm{SPEB}(w_{\mathrm r}(\cdot),w_{\mathrm t}(\cdot);\mathbf{p}) \\
			& \hspace{0.1cm} \mathrm{s.t.} && \sum_{\zeta\in\mathcal S} w_{\mathrm t}(\zeta)=1, \sum_{\zeta\in\mathcal S} w_{\mathrm r}(\zeta)=1, \label{eq. cons-distri-1}\\ 
			& && w_{\mathrm t}(\zeta), w_{\mathrm r}(\zeta)\ge 0,\quad \forall \zeta\in\mathcal S. \label{eq. cons-distri-2}
		\end{alignat}
	\end{subequations}
	Here, $\mathrm{SPEB}(w_{\mathrm r}(\cdot),w_{\mathrm t}(\cdot);\mathbf{p})$ denotes the SPEB evaluated under the distribution $w_{\mathrm r}(\cdot)$ and $w_{\mathrm t}(\cdot)$ for a target at $\mathbf p$.
	
	To solve $(\text{P1}')$, we adopt a two-stage approach to decouple the min-max structure. We first address the inner maximization stage by defining the worst-case design metric $F(w_{\mathrm r}(\cdot),w_{\mathrm t}(\cdot))$ as
	\begin{equation}
		\label{eq:def-worstcase-distribution}
		F(w_{\mathrm r}(\cdot),w_{\mathrm t}(\cdot))
		\;=\;
		\max_{\mathbf{p}\in\Gamma}
		\;\mathrm{SPEB}(w_{\mathrm r}(\cdot),w_{\mathrm t}(\cdot);\mathbf{p}).
	\end{equation}
	Then, the relaxed distributional problem is formulated as
	\begin{subequations}
		\begin{alignat}{2}
			(\text{P2}):\quad & \min_{w_{\mathrm r}(\cdot),\,w_{\mathrm t}(\cdot)} \quad && F(w_{\mathrm r}(\cdot),w_{\mathrm t}(\cdot)) \\
			& \hspace{0.6cm} \mathrm{s.t.} \quad && \eqref{eq. cons-distri-1}~\text{and}~\eqref{eq. cons-distri-2}. \nonumber
		\end{alignat}
	\end{subequations}
	\begin{proposition}
		\label{prop:active_symmetry}
		For any	distribution $w(\cdot)$ supported on $\mathcal S \subseteq [-a,a]$, define its symmetrized
		counterpart as
		\begin{equation}
			w_{\mathrm{sym}}(\zeta)
			=
			\frac{w(\zeta)+w(-\zeta)}{2},
			\qquad \forall\,\zeta\in\mathcal S.
		\end{equation}
		Then, for any pair $(w_{\mathrm r}(\cdot),w_{\mathrm t}(\cdot))$, the corresponding
		symmetrized distributions satisfy
		\begin{equation}
			F\!\left(
			w_{\mathrm r,\mathrm{sym}}(\cdot),
			w_{\mathrm t,\mathrm{sym}}(\cdot)
			\right)
			\leq
			F\!\left(w_{\mathrm r}(\cdot),w_{\mathrm t}(\cdot)\right).
		\end{equation}
		Hence, an optimal solution exists within the class of pairs of
		centro-symmetric distributions.
	\end{proposition}
	\begin{IEEEproof}
		Please see Appendix~\ref{prop.symmetry-mn}.
	\end{IEEEproof}
	Proposition~\ref{prop:active_symmetry} implies that an optimal solution follows the symmetric distribution\footnote{Note that the symmetry of $w_{\mathrm t}(\cdot)$ does not imply that the full transmit-antenna placement vector $\mathbf{x}_{\mathrm t}$ is centro-symmetric, since $w_{\mathrm t}(\cdot)$ is a power-weighted location distribution jointly determined by antenna positions $\mathbf{x}_{\mathrm t}$ and the power allocation vector $\bm{\rho}$.}. For such centro-symmetric distributions, the odd-moment term becomes zero, i.e., $m_1=m_3=s_1=s_3=0$. Consequently, for any centro-symmetric distributions $w_{\mathrm r,\mathrm{sym}}(\cdot)$ and $w_{\mathrm t,\mathrm{sym}}(\cdot)$, $\mathrm{SPEB}$ is reduced to
	\begin{equation}
		\begin{aligned} \label{eq:mono-SPEB-sym}
			\operatorname{SPEB}(w_{\mathrm r,\mathrm{sym}}(\cdot),w_{\mathrm t,\mathrm{sym}}(\cdot),\mathbf{p})
			=
			\kappa(r)\left(\frac{k_1}{\Xi_1}+\frac{k_2}{\Xi_2}\right).
		\end{aligned}
	\end{equation}
	where $\Xi_1=m_2+s_2$, $\Xi_2=(m_4-m_2^2)+(s_4-s_2^2)$, $k_1 = \frac{r^2(1+3u^2)}{(1-u^2)^2}$, and $
	k_2 = \frac{4r^4}{(1-u^2)^2}$. 
	Define $\mathbf{p}_{\mathrm{worst}}$ as a worst-case target location for centro-symmetric distributions,
	i.e.,
	\begin{equation}
		\mathbf{p}_{\mathrm{worst}}
		\in
		\arg\max_{\mathbf{p}\in\Gamma} \quad \mathrm{SPEB}(w_{\mathrm r,\mathrm{sym}}(\cdot),w_{\mathrm t,\mathrm{sym}}(\cdot), \mathbf{p}).
		\label{eq:worst_location}
	\end{equation}
	\begin{proposition}
		\label{prop:active_broadside}
		For any centro-symmetric distributions $w_{\mathrm r,\mathrm{sym}}(\cdot)$ and $w_{\mathrm t,\mathrm{sym}}(\cdot)$ supported on $\mathcal S \subseteq [-a,a]$, 
		the worst-case
		target location is independent of $(w_{\mathrm r,\mathrm{sym}}(\cdot),w_{\mathrm t,\mathrm{sym}}(\cdot))$, and is
		given by
		\begin{equation}
			\begin{aligned}
				\mathbf{p}_{\mathrm{worst}}
				= [\,0,\; d_{\mathrm{R},\max}\,]^{\mathrm{T}},
			\end{aligned}
		\end{equation}
		where $d_{\mathrm{R},\max}=\frac{2D^2}{\lambda}$ denotes the broadside Rayleigh distance.
	\end{proposition}
	\begin{IEEEproof}
		Please see Appendix~\ref{prop.worstpoint-mn}.
	\end{IEEEproof}
    Therefore, the variable $\kappa(d_{\mathrm{R},\max})$ is a positive constant independent
	of the antenna distributions, which is further denoted as
	$\tilde{\kappa}=\kappa(d_{\mathrm{R},\max})$. Accordingly, variables $k_1$ and $k_2$ in \eqref{eq:mono-SPEB-sym}
	evaluated at the worst-case location are
	$d_{\mathrm{R},\max}^{2}$ and $4d_{\mathrm{R},\max}^{4}$,
	respectively.
	Accordingly, problem (P2) reduces to
	\begin{subequations} \label{eq:opt-sym-broadside-mn}
		\begin{alignat}{2}
			(\text{P3}):\quad & \min_{w_{\mathrm r,\mathrm{sym}}(\cdot),\,w_{\mathrm t,\mathrm{sym}}(\cdot)} \quad && \tilde{\kappa}\left(\frac{d_{\mathrm{R},\max}^2}{\Xi_1}+\frac{4d_{\mathrm {R},\max}^4}{\Xi_2}\right). \\
			& \hspace{1cm} \mathrm{s.t.} \quad && \eqref{eq. cons-distri-1}~\text{and}~\eqref{eq. cons-distri-2}. \nonumber
		\end{alignat}
	\end{subequations}
	It is worth noting that the objective function of (P3) depends on $w_{\mathrm r,\mathrm{sym}}(\cdot)$ and $w_{\mathrm t,\mathrm{sym}}(\cdot)$ solely through $\Xi_1$ and $\Xi_2$, i.e., through the aggregated second and fourth order moment terms associated with $X_{\mathrm r}$ and $X_{\mathrm t}$. This explicit moment dependence plays a crucial role in the subsequent structural analysis.

	\subsection{Moment-Based Optimal Distribution and Practical Realization}
	Let $\mathcal{M}_{\mathrm{r}} = \bigl\{m_0,m_2,m_4\bigr\}$ and $\mathcal{M}_{\mathrm{r}} = \bigl\{s_0,s_2,s_4\bigr\}$ denote the attainable moment sets w.r.t. $X_{\mathrm r}$ and $X_{\mathrm t}$, respectively. By inspecting the structure of (P3), we establish the following
	two lemmas for the optimization over $\mathcal{M}_{\mathrm{r}}$ and $\mathcal{M}_{\mathrm{t}}$.

	\begin{lemma}
		The minimum of $\text{(P3)}$ over moment sets $\mathcal{M}_{\mathrm{r}}$ and $\mathcal{M}_{\mathrm{t}}$ can be attained by centro-symmetric discrete distributions $w_{\mathrm t,\mathrm{sym}}(\cdot)$ and $w_{\mathrm r,\mathrm{sym}}(\cdot)$ supported on at most three points of the given aperture.\footnote{Any such centro-symmetric
			three-point distribution must include a location at the origin.}
	\end{lemma}
	
	\begin{IEEEproof}
		Please see Appendix~\ref{app.lemma3}.
	\end{IEEEproof}
	
	\begin{lemma}
		The minimum of $\text{(P3)}$ over $\mathcal{M}_{\mathrm{r}}$ and $\mathcal{M}_{\mathrm{t}}$ can be achieved by centro-symmetric discrete distributions supported at the center and two edges of the given aperture, i.e., $\{-a,0,+a\}$.
	\end{lemma}
	
	\begin{IEEEproof}
		Please see Appendix~\ref{app.lemma4}.
	\end{IEEEproof}
	Consider such three-point centro-symmetric distributions $w_{\mathrm r}(\cdot)$
	and $w_{\mathrm t}(\cdot)$ supported on $\{-a,0,a\}$ with corresponding probability masses $\left(\frac{q_{\mathrm r}}{2},\,1-q_{\mathrm r},\,\frac{q_{\mathrm r}}{2}\right)
	$ and $\left(\frac{q_{\mathrm t}}{2},\,1-q_{\mathrm t},\,\frac{q_{\mathrm t}}{2}\right)$ respectively, where $q_{\mathrm r}, q_{\mathrm t} \in (0,1)$.
	Hence, we obtain 
	\begin{equation}
		\label{eq: moments}
		m_2 = q_{\mathrm r} a^2,
		\quad
		m_4 = q_{\mathrm r} a^4,
		\quad
		s_2=q_{\mathrm t} a^2,
		\quad
		s_4=q_{\mathrm t} a^4.
	\end{equation} By
	substituting \eqref{eq: moments} into (P3), the optimization problem becomes
	\begin{equation}
		\begin{aligned}
			\label{eq:active_J_qt_qr}
			\min_{q_{\mathrm t},q_{\mathrm r} \in (0,1)} \quad \tilde{\kappa} \left(\frac{4d_{\mathrm {R},\max}^4}
			{a^4\!\left[q_{\mathrm r}(1-q_{\mathrm r})+q_{\mathrm t}(1-q_{\mathrm t})\right]}
			+
			\frac{d_{\mathrm{R},\max}^2}
			{a^2(q_{\mathrm r}+q_{\mathrm t})}\right).	
		\end{aligned}
	\end{equation}
	By inspecting \eqref{eq:active_J_qt_qr}, we further obtain the following lemma:
	\begin{lemma}
		\label{lem:qtqr_active}
		Every minimizer of \eqref{eq:active_J_qt_qr}, denoted by $q_{\mathrm r}^\star$ and $q_{\mathrm t}^\star$, satisfies
		$q_{\mathrm r}^\star=q_{\mathrm t}^\star$.
	\end{lemma}
	\begin{IEEEproof}
		Please see Appendix~\ref{app.lem:qtqr_active}.
	\end{IEEEproof}
	Accordingly, for any minimizer, let $q = q_{\mathrm r}^\star=q_{\mathrm t}^\star$.
	By Proposition~1, the search for optimal distributions $w_{\mathrm t}(\cdot)$ and $w_{\mathrm r}(\cdot)$ is restricted to
	centro-symmetric distributions. Furthermore, Lemma~3 and Lemma~4 imply that, within the
	class of centro-symmetric distributions, the optimal distributions $w_{\mathrm t}(\cdot)$ and $w_{\mathrm r}(\cdot)$ can both be chosen to have support on $\{-a,0,+a\}$. This yields the
	following theorem.
	
	\begin{theorem}
		\label{thm:three-point-optimal-mn}
		The optimal distributions $w_{\mathrm t}(\cdot)$ and $w_{\mathrm r}(\cdot)$ of (P3) can be obtained by a three-point distribution supported on $\{-a,0,+a\}$ and share the same optimal probability masses $\left(\frac{q}{2},\,1-q,\,\frac{q}{2}\right)$.
	\end{theorem}
	
	Therefore, problem \eqref{eq:active_J_qt_qr} reduces to 
	\begin{equation}
		\begin{aligned}
			\label{eq:active_J_q}
			\min_{q \in (0,1)} \quad \tilde{\kappa} \left(\frac{4d_{\mathrm {R},\max}^4}
			{2a^4q(1-q)}
			+
			\frac{d_{\mathrm{R},\max}^2}
			{2a^2q}\right).
		\end{aligned}
	\end{equation}
	This is a convex optimization problem, and the unique minimizer $q^\star\in(\frac{1}{2},1)$ can be obtained by taking the derivative with respect to $q$, yielding
	\begin{equation}
		\begin{aligned}
			\label{eq:qstar}
			&q^\star \;=\; 1 + \gamma - \sqrt{\gamma(1+\gamma)}.
		\end{aligned}
	\end{equation}
	where $\gamma = \frac{4d_{\mathrm{R},\max}^2}{a^2} = \frac{256a^2}{\lambda^2}$.
	It is worth noting that $q^\star$ monotonically decreases from $1$ and asymptotically approaches $0.5$ as the value of $a/\lambda$ increases, as shown in Fig.~\ref{fig:pstar}. In practical deployments where the aperture spans several wavelengths ($a \gg \lambda$), $q^\star$ rapidly converges to its asymptotic limit. Hence, we can directly set $q^\star = 0.5$, yielding a fixed probability mass allocation of $(0.25, 0.5, 0.25)$ for an efficient near-optimal distribution design.
	\begin{figure}[t]
		\centering
		\includegraphics[width=0.75\linewidth, keepaspectratio]{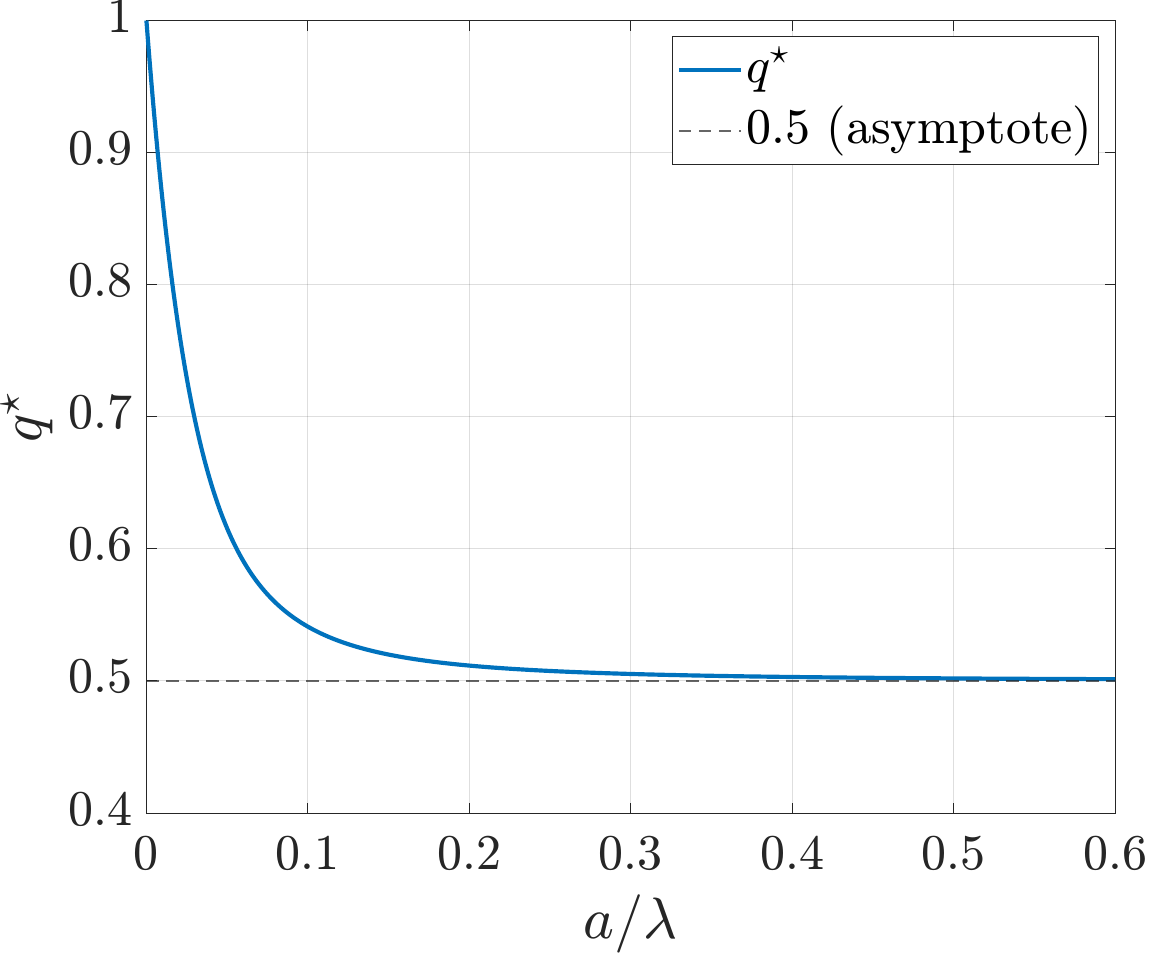}
		\caption{Optimal $q^\star$ versus $a/\lambda$.}
		\label{fig:pstar}
		\vspace{-5mm}
	\end{figure}
	Having characterized the closed-form optimal solution for problem $(\text{P1}')$, we now incorporate the inter-element spacing constraint and the given number of transmit antennas $M$ and receive antennas $N$ to ensure practical feasibility. Based on the optimal three-point centro-symmetric distribution obtained above, we propose a discrete realization strategy that approaches the theoretical optimum in $(\text{P1}')$ while satisfying the minimum spacing $d$.
	
	Specifically, on the receiver side, $w_{\mathrm r}(\cdot)$ only characterizes the geometric distribution of the receive antennas. Therefore, its practical implementation is solely determined by the placement of $\mathbf{x}_{\mathrm r}$. Moreover, as indicated by \eqref{eq.FIM-monostatic}, all receive antennas contribute to the receive-side Fisher information. Hence, the receiver-side deployment should fully utilize all $N$ receive antennas. To approximate the optimal receive-side distribution while maintaining the minimum inter-element spacing, a simple yet effective strategy is to allocate $N_{\mathrm{left}} = N_{\mathrm{right}} = \operatorname{round}(0.25N)$ antennas to form clusters at the left and right aperture endpoints, respectively, while the remaining $N_{\mathrm{center}} = N - N_{\mathrm{left}} - N_{\mathrm{right}}$ antennas are placed around the aperture center. Within each of these three clusters (i.e., the center and two edges), the antennas are placed uniformly with the fixed minimum spacing $d$.
	
	On the transmitter side, $w_{\mathrm t}(\cdot)$ is a power-weighted transmitter-location distribution, and hence is jointly determined by the transmit locations $\mathbf{x}_{\mathrm t}$ and the power allocation vector $\bm{\rho}$. One realization of the optimal three-point distribution $w_{\mathrm t}(\cdot)$ is to allocate the prescribed power $P$ to only three transmit antennas located at $\{-a,0,a\}$ with transmit powers $P_{\mathrm{left}} = P_{\mathrm{right}} = 0.25P$ and $P_{\mathrm{center}} = 0.5P$, while keeping the remaining transmit antennas inactive. This avoids the performance loss induced by minimum inter-element spacing constraints, thereby attaining the theoretical optimum at the transmitter side.
	\begin{remark} \label{remark:mono_tx}
		The above analysis reveals that the optimal transmit-side distribution
		$w_{\mathrm t}^\star(\cdot)$ can be exactly realized by activating only three
		antennas located at $-a$, $0$, and $a$, with transmit powers
		$\frac{q^\star}{2}P$, $(1-q^\star)P$, and
		$\frac{q^\star}{2}P$, respectively. Consequently, deploying more than three transmit antennas yields no further reduction in the optimal SPEB, as their sheer number becomes mathematically irrelevant once the optimal spatial distribution is established. From a hardware perspective, this minimal three-antenna architecture drastically reduces the required number of active RF chains, thereby minimizing both hardware cost and implementation complexity of the system.
	\end{remark}
	\vspace{-4mm}
	\section{Numerical Results}
	This section presents simulation results to validate the correctness of SPEB derivation and evaluate the effectiveness of the proposed joint transmit and receive antenna placements and the transmit power allocation.
	The simulation parameters are set as follows. The carrier frequency is $f_0 = 28~\text{GHz}$ \cite{ZhangTWC2022}, and the number of snapshots is $T = 1024$. Unless otherwise specified, the half-aperture $a = 25\lambda$, minimum inter-antenna spacing $d=\frac{\lambda}{2}$ \cite{YuanOJAP2023}, number of receive antennas $N = 25$, number of transmit antennas $M=12$, and average received $\text{SNR} = 5~\text{dB}$ \cite{JoJSAC2023}.
	\subsection{MSE versus SPEB}
	To validate the derived CRB and SPEB, we evaluate the MSE of a practical MLE \cite{BekkermanTSP}. Specifically, we consider a representative target located at $(u_0, r_0) = (0.2, 16\,\mathrm{m})$ to compare the resulting MSE obtained by MLE against the derived SPEB under the proposed transmit and receive array geometries and power allocation.
	\begin{figure}[t!]
		\centering
		\includegraphics[width=0.75\linewidth]{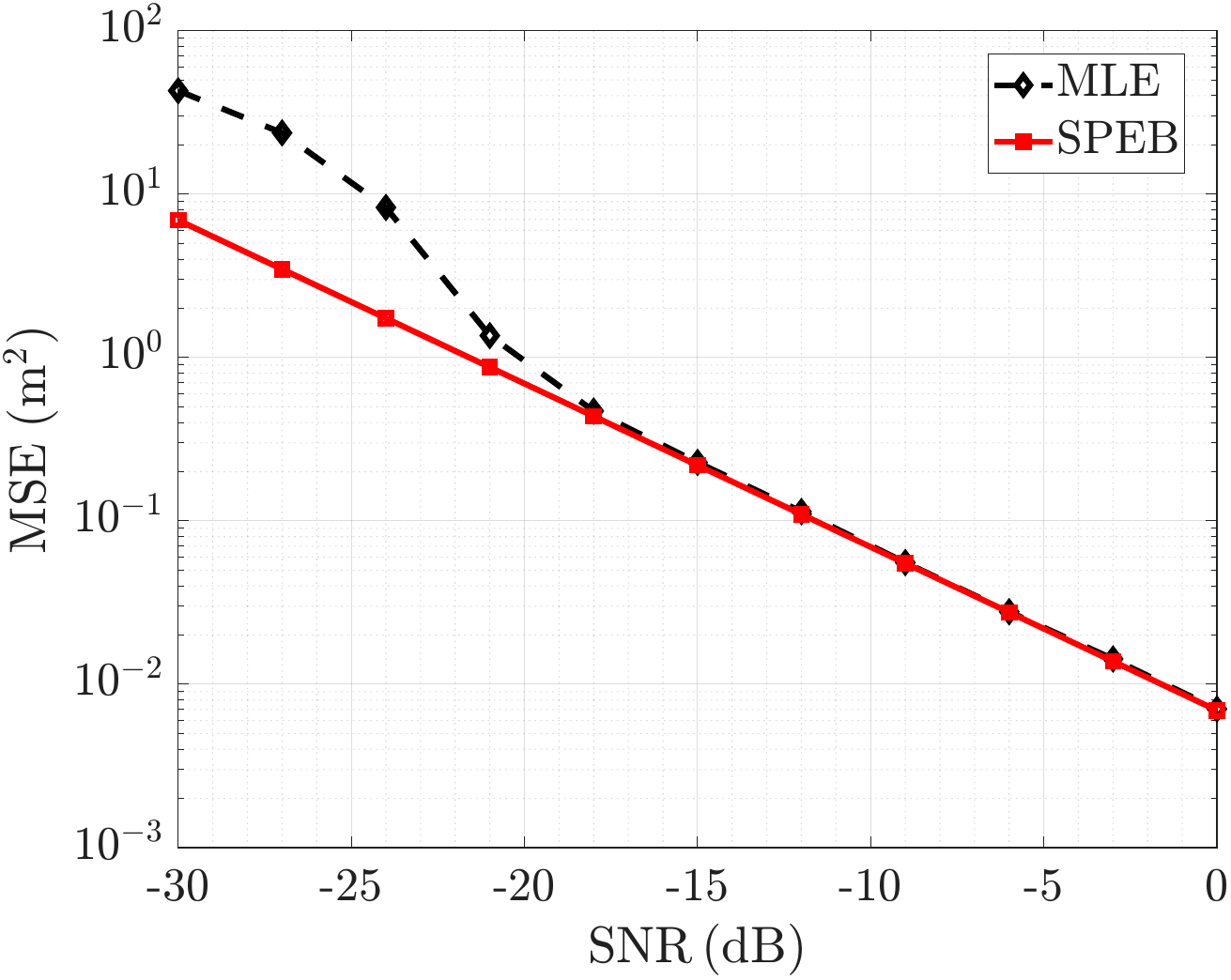}
		\caption{MSE of the MLE and the derived
			SPEB versus SNR under the proposed design.}
		\label{fig:MSE_comparison}
		\vspace{-2mm}
	\end{figure}
	Fig.~\ref{fig:MSE_comparison} illustrates the MSE of the MLE and the derived SPEB versus the received SNR. As observed, the SPEB decreases linearly with the SNR in dB, which mathematically reflects its inverse proportionality to the linear SNR, as indicated by the scaling factor $\kappa(r) = \frac{\sigma^2\lambda^2}{8\pi^2TPN|\alpha|^2\beta^2(r)}$ in \eqref{eq.FIM-monostatic}. Moreover, in the low-SNR regime, the MLE exhibits a noticeable performance gap from the SPEB, while in the high-SNR regime, its MSE closely matches the SPEB. This is consistent with \cite{kay1993fssp_estimation} and thereby validates the correctness of our CRB and SPEB derivations.
	\subsection{Accuracy of the Fresnel Approximation}
	To validate the accuracy of the Fresnel approximation adopted
	in \eqref{eq:steering_fresnel}, we compare the resulting SPEB with that
	numerically evaluated using the exact spherical-wave propagation distances
	in $r_m^{(\mathrm{t})}(\bm{\eta})$ and $r_n^{(\mathrm{r})}$. Specifically, we fix the target direction
	at the array broadside, i.e., $u=0$, and vary the target range $r$ over the
	radiating near-field region $[d_{\mathrm F},d_{\mathrm R}]$. The proposed transmit and receive antenna placements and transmit power allocation are adopted for both cases.
	We define the relative SPEB error as
	\begin{equation}
		\epsilon_{\mathrm{SPEB}}
		=
		\frac{
			\left|
			\mathrm{SPEB}_{\mathrm{exact}}
			-
			\mathrm{SPEB}_{\mathrm{Fresnel}}
			\right|
		}{
			\mathrm{SPEB}_{\mathrm{exact}}
		}
		\times 100\%,
		\label{eq:relative_SPEB_error}
	\end{equation}
	where $\mathrm{SPEB}_{\mathrm{exact}}$ and
	$\mathrm{SPEB}_{\mathrm{Fresnel}}$ denote the SPEBs obtained using the
	exact spherical-wave model and the Fresnel approximation,
	respectively.
	\begin{figure}[t!]
		\centering
		\includegraphics[width=0.75\linewidth]{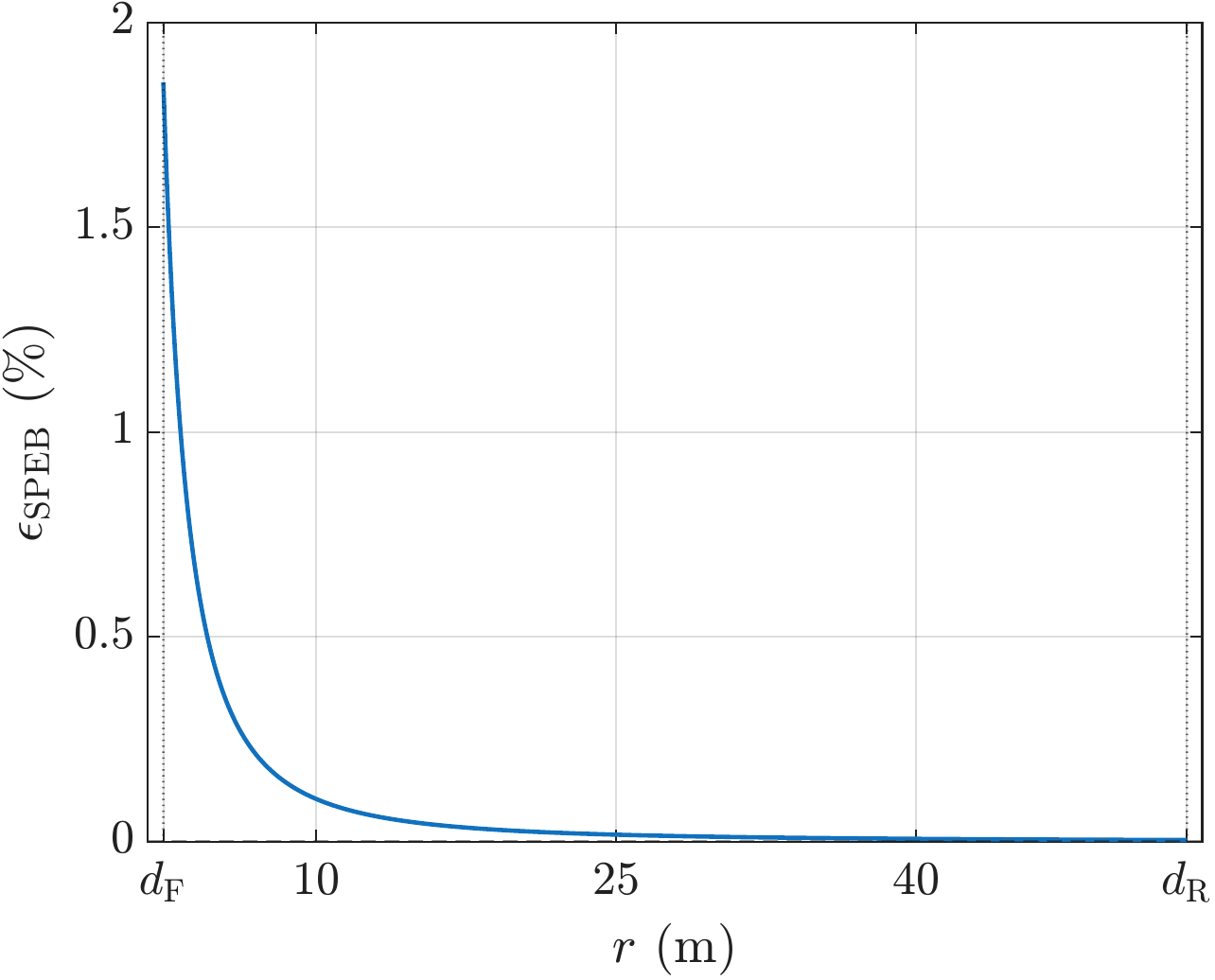}
		\caption{Relative SPEB error between the exact spherical-wave model
			and the Fresnel approximation versus the target range
			at $u=0$.}
		\label{fig:Fresnel_validation}
		\vspace{-2mm}
	\end{figure}
	Fig.~\ref{fig:Fresnel_validation} shows that the relative SPEB error is
	approximately $1.85\%$ at the Fresnel boundary
	$d_{\mathrm F}=2.35~\mathrm{m}$ and decreases rapidly as the target range
	increases. In particular, the error falls below approximately $0.1\%$ for
	$r\gtrsim10~\mathrm{m}$ and becomes negligible near the Rayleigh boundary
	$d_{\mathrm R}=53.53~\mathrm{m}$. This trend follows from the broadside
	distance expansion
	$
	\sqrt{r^2+x^2}-r
	=
	\frac{x^2}{2r}
	-
	\frac{x^4}{8r^3}
	+
	\mathcal{O}\left(\frac{x^6}{r^5}\right),
	$
	where the dominant neglected term decays with $r^{-3}$. These results
	demonstrate that the adopted Fresnel approximation provides
	an accurate characterization of the SPEB over the considered
	near-field region.
	
	\subsection{Performance Evaluation of Power Allocation}
	We next evaluate the proposed power allocation strategy under two fixed transmit and receive antenna placements, i.e., the proposed three-point-based antenna placement and the uniform placement. The latter is defined as follows:
	\begin{itemize}
		\item {\textbf{Uniform Placement} \cite{ShiTSP2021}:} 
		The transmit and receive antennas are uniformly placed over the entire aperture $[-a,a]$, with inter-element spacings $D/(M-1)$ and $D/(N-1)$, respectively.
	\end{itemize}
	Under each antenna placement, we compare the worst-case SPEB achieved by the proposed power allocation strategy with the following benchmark schemes:
	\begin{itemize}
		\item \textbf{Uniform power} \cite{WangTSP0124}: The total transmit power is distributed equally among all transmit antennas, i.e., $\rho_m=1/M$ for all $m$.
		
		\item \textbf{Two-edge power} \cite{WerfTSP2026}: The total transmit power is equally allocated to the two transmit antennas located at the aperture edges.
		
		\item \textbf{Three-point equal power:} The total transmit power is equally allocated to the three transmit antennas located at the two aperture edges and at, or closest to, the aperture center.
	\end{itemize}
	\begin{figure}[t]
		\centering
		\subfloat[Uniform placement\label{fig:speb-G1-M}]{
			\includegraphics[width=0.75\linewidth]{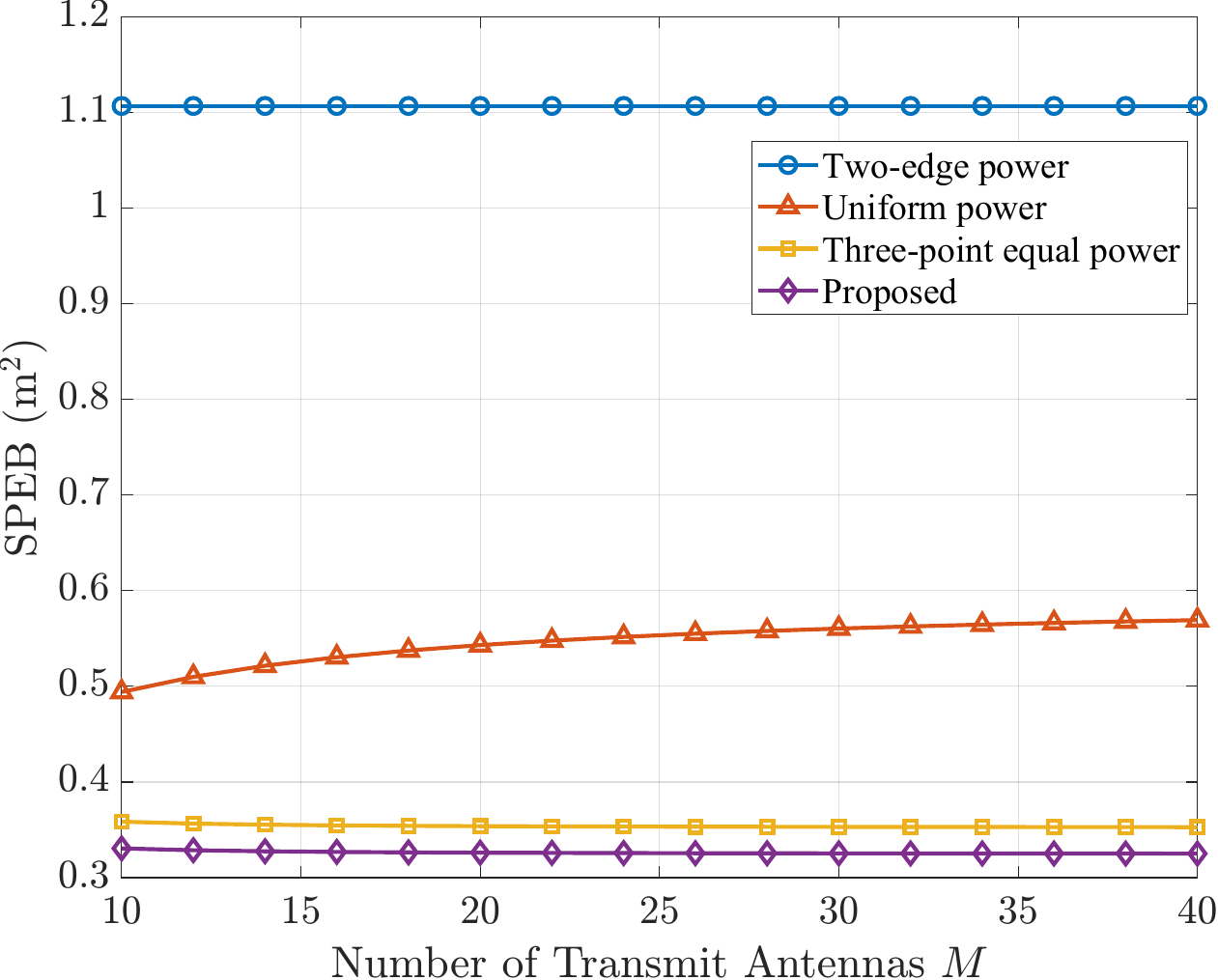}
		}
		\\
		\subfloat[Proposed placement\label{fig:speb-G2-M}]{
			\includegraphics[width=0.75\linewidth]{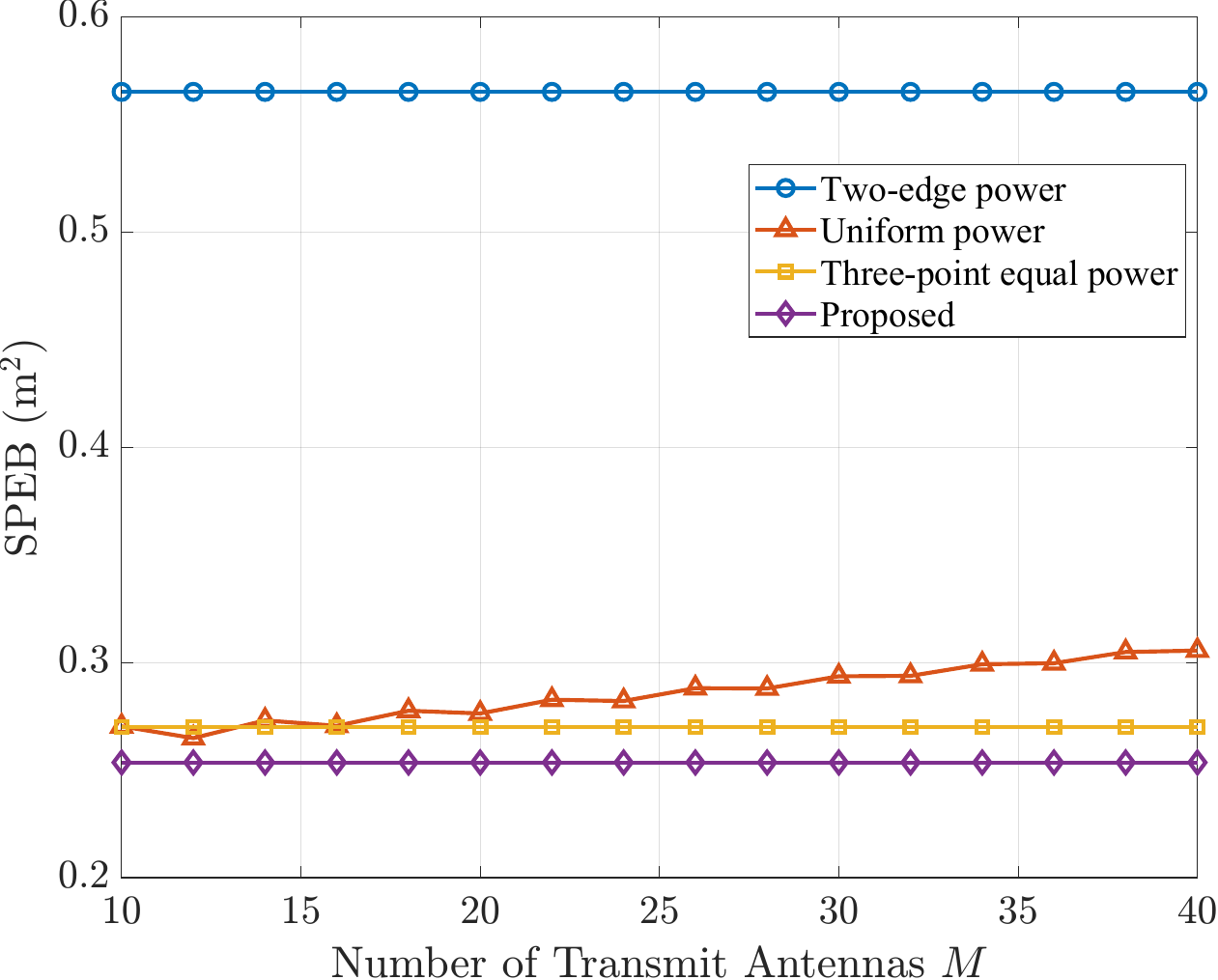}
		}
		\caption{Worst-case SPEB versus the number of transmit antennas $M$ under different antenna placements.}
		\label{fig:speb-G-M}
	\end{figure}
	
	Fig.~\ref{fig:speb-G-M} shows the worst-case SPEB versus the number of transmit antennas $M$ under various power allocation schemes for both the uniform antenna placement (Fig.~\ref{fig:speb-G-M}(a)) and the proposed antenna placement (Fig.~\ref{fig:speb-G-M}(b)). As observed, the proposed transmit power allocation achieves the lowest worst-case SPEB compared to all benchmark schemes across various numbers of transmit antennas, validating the effectiveness of the derived three-point power allocation strategy.
	Moreover, across various values of $M$, the worst-case SPEB obtained in Fig.~\ref{fig:speb-G-M}(b) with proposed antenna placement is generally lower than that in Fig.~\ref{fig:speb-G-M}(a) with uniform placement. This also proves the effectiveness of our proposed antenna placements.
	Furthermore, it can be observed that among the proposed, two-edge, and three-point equal power schemes, the worst-case SPEB remains unchanged as $M$ increases. This is consistent with \eqref{eq:Mt1}-\eqref{eq:Mt3}, which show that the transmit-side contribution to SPEB is determined by the power-weighted distribution $w_{\mathrm t}(\cdot)$, rather than the sheer number of transmit antennas. Since these three schemes activate only a fixed subset of transmit antennas, increasing $M$ does not alter the power-weighted spatial distribution, causing the worst-case SPEB to remain constant. 
	Conversely, the uniform power scheme exhibits a gradual performance degradation as $M$ increases. This is because the total transmit power is spread over all transmit antennas, thereby allocating more power to suboptimal positions that deviate from the derived optimal locations, i.e., edges and the center, causing the power-weighted distribution $w_{\mathrm t}(\cdot)$ to deviate from the optimal three-point distribution, ultimately leading to a higher worst-case SPEB.
	These results align with Remark~\ref{remark:mono_tx}, underscoring that simply increasing the number of transmit antennas does not necessarily improve sensing performance. Instead, an appropriate power-weighted spatial distribution is essential for attaining the theoretically optimal SPEB.
	\begin{figure}[t]
		\centering
		\subfloat[Uniform placement\label{fig:speb-G1-N}]{
			\includegraphics[width=0.75\linewidth]{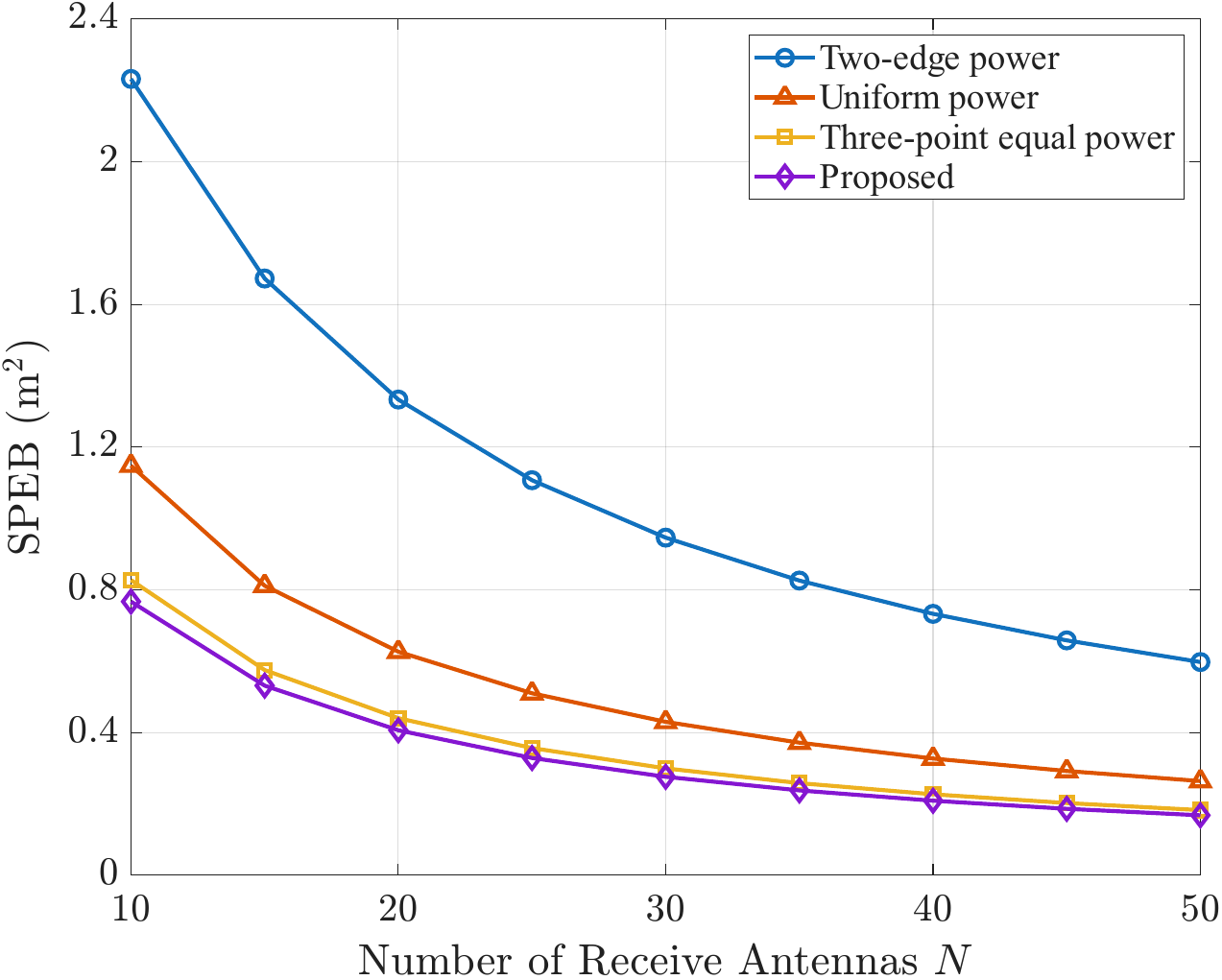}
		}
		\\ 
		\subfloat[Proposed placement\label{fig:speb-G2-N}]{
			\includegraphics[width=0.75\linewidth]{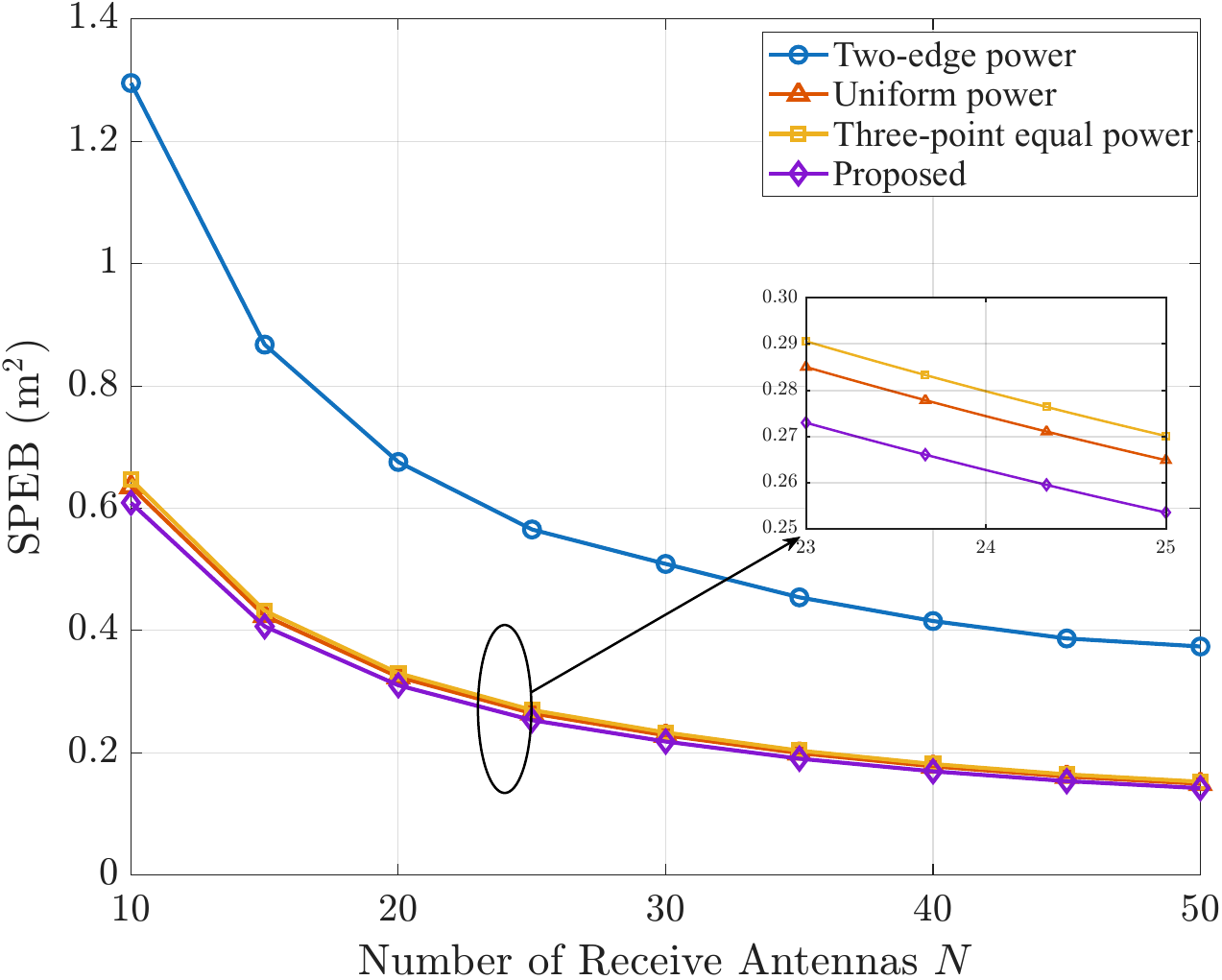}
		}
		\caption{Worst-case SPEB versus the number of receive antennas $N$ under different antenna placements.}
		\label{fig:speb-G-N}
	\end{figure}
	
	Fig.~\ref{fig:speb-G-N} shows the worst-case SPEB versus the number of receive antennas $N$ under the two antenna placement settings described above. It is observed that the proposed power allocation consistently achieves the lowest worst-case SPEB among all benchmark schemes. This result further proves the superiority of the proposed power-allocation design. Additionally, the worst-case SPEB decreases for all schemes as $N$ increases. This trend agrees with the scaling factor $\kappa(r)$ in \eqref{eq.FIM-monostatic}, which is inversely proportional to $N$. Therefore, adding receive antennas directly reduces the overall SPEB. This behavior is fundamentally different from the transmit-side trend in Fig.~\ref{fig:speb-G-M}, as every additional receive antenna provides an independent spatial observation that directly suppresses the SPEB, while the transmit antennas contribute to the sensing performance solely through their power-weighted distribution, rendering their sheer number irrelevant once the optimal distribution is established.
	\subsection{Performance of Joint Transmit/Receive Antenna Placements and Power Allocation Design}
	To further validate the effectiveness of the proposed joint design, we compare its performance against the following benchmark schemes:
	\begin{itemize} 
		\item \textbf{Sparse ULA} \cite{WangTSP0124}: Both the transmit and receive arrays are fixed as sparse ULAs spanning the entire aperture, under uniform power allocation.
		\item \textbf{Two-edge} \cite{WerfTSP2026}: The receive antennas are clustered at the two aperture edges with minimum inter-element spacing, while the transmit power is equally allocated to two active transmit antennas located at the aperture edges.
		\item \textbf{Opt-Tx:} The receive array is fixed as a sparse ULA, while the transmit antenna positions are optimized via exhaustive search over the aperture under uniform power allocation.
		\item \textbf{Opt-Rx:} The transmit array is fixed as a sparse ULA under uniform power allocation, while the receive antenna positions are optimized via exhaustive search.
		\item \textbf{Joint exhaustive search:} The transmit/receive antenna positions are jointly optimized via exhaustive
		search, while fixing the transmit-power fractions to the proposed closed-form
		values, i.e., $q^\star/2$, $1-q^\star$, and $q^\star/2$ for the
		leftmost, middle, and rightmost transmit antennas, respectively.
	\end{itemize}
	\begin{figure}[t]
		\centering
		\includegraphics[width=0.75\linewidth]{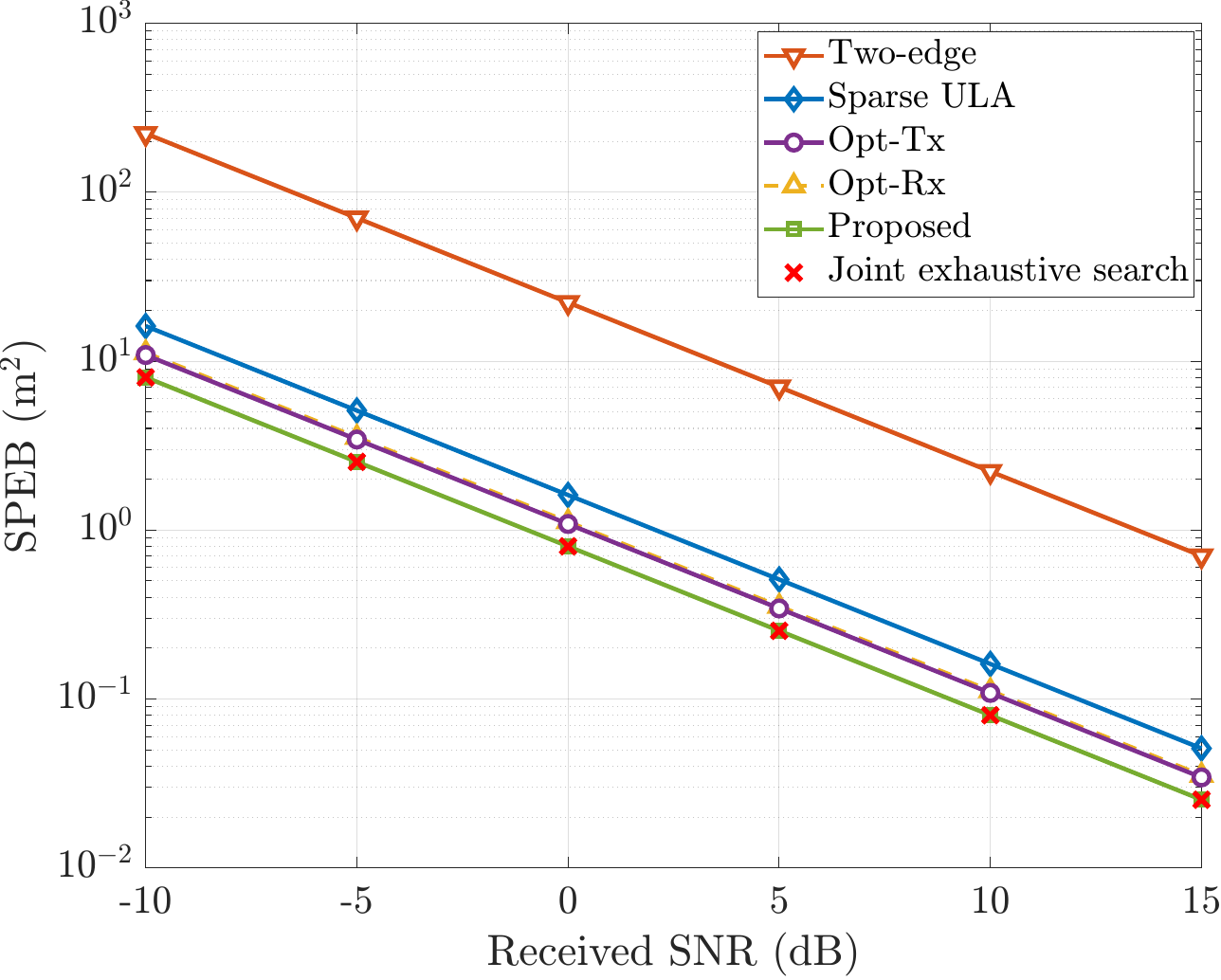}
		\caption{Worst-case $\text{SPEB}$ versus received SNR.}
		\label{fig:joint_SNR}
	\end{figure}
	\begin{figure}[t]
		\centering
		\includegraphics[width=0.75\linewidth]{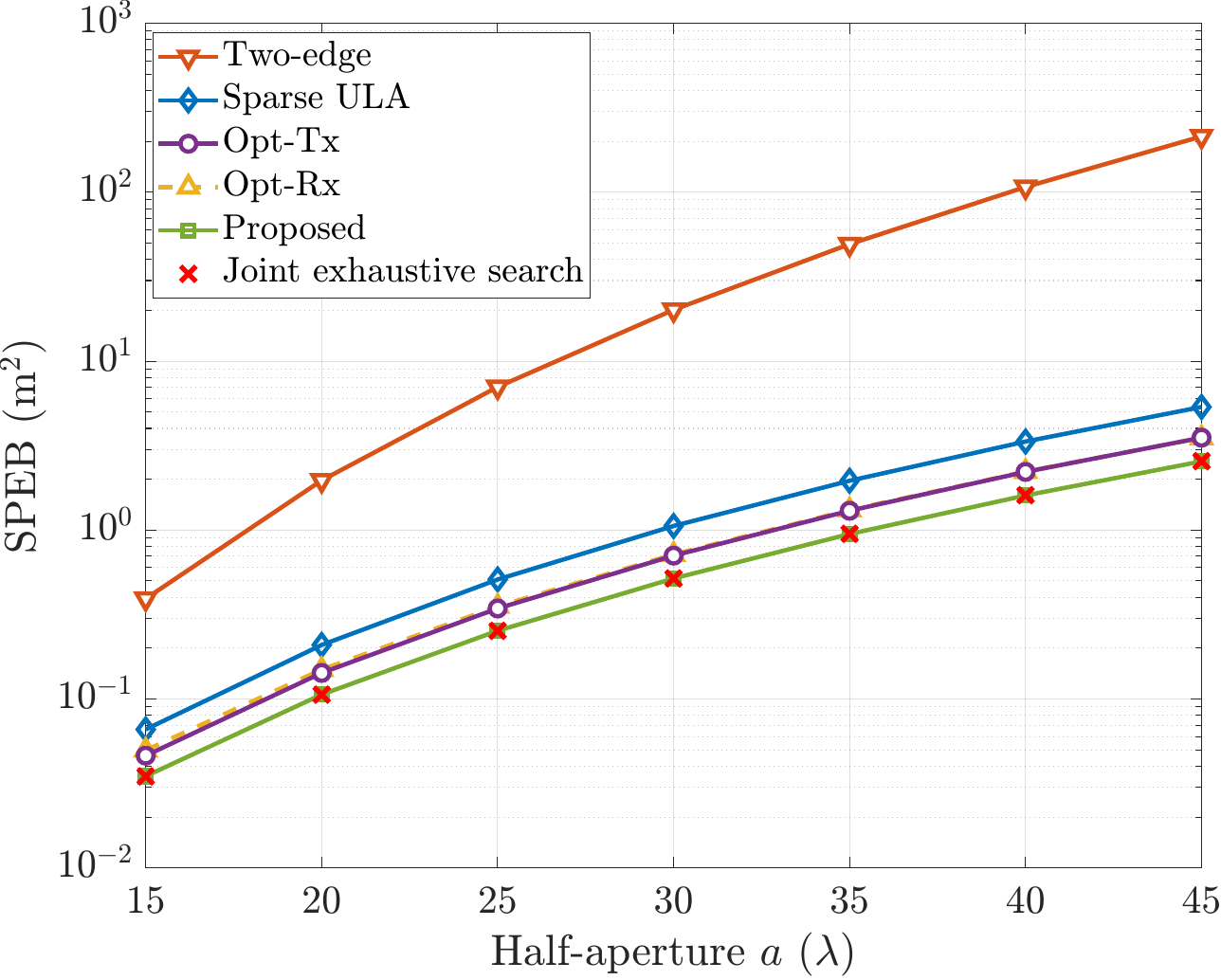}
		\caption{Worst-case $\text{SPEB}$ versus aperture size.}
		\label{fig:joint_design-aperture}
	\end{figure}
	
	Fig.~\ref{fig:joint_SNR} shows the worst-case SPEB versus the received SNR. As expected, the SPEB of all schemes decreases inversely with the received SNR, which is consistent with the SNR dependence of $\kappa(r)$ in \eqref{eq.FIM-monostatic}. The two-edge benchmark yields the largest SPEB, followed by the sparse ULA benchmark. By optimizing either the transmit or receive geometry, Opt-Tx and Opt-Rx achieve lower SPEB than the sparse ULA benchmark. Nevertheless, both schemes remain inferior to the proposed design, since they optimize only one-side array geometry and retain uniform transmit power allocation. In contrast, the proposed design jointly exploits the transmit/receive antenna placements and the transmit power allocation, leading to the lowest worst-case SPEB. Moreover, its performance closely matches that of the joint exhaustive-search benchmark over the entire SNR range, indicating the effectiveness of the proposed closed-form design with substantially lower computational complexity.
	
	Fig.~\ref{fig:joint_design-aperture} illustrates the worst-case SPEB against the half-aperture size $a$. It shows that our proposed design consistently achieves the lowest SPEB among all evaluated baselines and closely matches the joint exhaustive-search benchmark over the entire aperture
	range, further demonstrating its near-optimality. Besides, we observe that for all schemes, the worst-case SPEB increases as the aperture grows. This trend is mainly attributed to the aperture-dependent worst-case location. By Proposition~\ref{prop:active_broadside}, the worst-case target lies at the broadside Rayleigh boundary, i.e., \(d_{\mathrm{R},\max}=8a^2/\lambda\). Therefore, increasing the half-aperture quadratically enlarges the worst-case range, which intuitively makes the target position more difficult to estimate.
	\vspace{-5mm}
	\section{Conclusion}
	\vspace{-5mm}
	In this paper, we investigated the joint design of transmit and receive antenna placements and the transmit power allocation for monostatic near-field sensing. We first derived the closed-form CRB and SPEB, then formulated a joint design problem aiming to minimize the worst-case SPEB. By temporarily relaxing the minimum inter-spacing constraint, we established a three-point optimal structure supported at the two aperture edges and the center based on structural analysis and moment-based methods. The results shows that the optimal transmit-side distribution can be realized by activating only three transmit antennas with the derived power allocation, while an efficient clustered receive-antenna deployment was developed to satisfy the spacing constraint. Numerical results showed that the proposed design consistently outperformed conventional uniform position array geometries and power-allocation schemes with negligible computational complexity.
	\appendices
	\vspace{-4mm}
	\section{Derivation of the CRB Matrix}
	\vspace{-4mm}
	\label{appendix:FIM_MN_monostatic}
	Let $\tilde{\mathbf{g}} = \beta(r)\mathbf{g}$, then the observation model in \eqref{eq:vectorized_statistic} becomes $\tilde{\mathbf{y}}=\alpha\tilde{\mathbf{g}}+\tilde{\mathbf{n}}$.
	According to \cite{BoyerTSP, WangTSP0124}, the CRB matrix for estimating $\bm\eta$ can be expressed as
	$
	\mathrm{CRB}_{\bm\eta}
	=
	\mathbf J^{-1}
	=
	\frac{\sigma^2}{2|\alpha|^2}(\mathbf Q')^{-1}
	$, 
	where 
	\begin{equation}
		{\mathbf Q}'
		=
		\begin{bmatrix}
			\|\tilde{\mathbf g}_u\|^2
			-
			\frac{|\tilde{\mathbf g}^{\mathrm H}\tilde{\mathbf g}_u|^2}
			{\|\tilde{\mathbf g}\|^2},
			&
			\Re\!\left\{
			\tilde{\mathbf g}_u^{\mathrm H}\tilde{\mathbf g}_r
			-
			\frac{
				(\tilde{\mathbf g}^{\mathrm H}\tilde{\mathbf g}_r)
				(\tilde{\mathbf g}_u^{\mathrm H}\tilde{\mathbf g})
			}
			{\|\tilde{\mathbf g}\|^2}
			\right\}
			\\
			\Re\!\left\{
			\tilde{\mathbf g}_u^{\mathrm H}\tilde{\mathbf g}_r
			-
			\frac{
				(\tilde{\mathbf g}^{\mathrm H}\tilde{\mathbf g}_r)
				(\tilde{\mathbf g}_u^{\mathrm H}\tilde{\mathbf g})
			}
			{\|\tilde{\mathbf g}\|^2}
			\right\}
			&
			\|\tilde{\mathbf g}_r\|^2
			-
			\frac{|\tilde{\mathbf g}^{\mathrm H}\tilde{\mathbf g}_r|^2}
			{\|\tilde{\mathbf g}\|^2}
		\end{bmatrix}.
		\label{eq:appA_Qprime}
	\end{equation}
	with $\tilde{\mathbf g}_u = \frac{\partial \tilde{\mathbf g}}{\partial u}$ and $\tilde{\mathbf g}_r = \frac{\partial \tilde{\mathbf g}}{\partial r}$, 
	and matrix $\mathbf J = \begin{bmatrix}
		J_{uu} & J_{ur}\\
		J_{ru} & J_{rr}\\
	\end{bmatrix} = \frac{2|\alpha|^2}{\sigma^2}\mathbf Q'$.
	$\mathrm{CRB}_u$ and $\mathrm{CRB}_r$ are given by
	\begin{equation}
		\begin{aligned}
			\mathrm{CRB}_u
			=
			\frac{J_{rr}}{J_{uu}J_{rr}-J_{ur}^2},
			\quad
			\mathrm{CRB}_r
			=
			\frac{J_{uu}}{J_{uu}J_{rr}-J_{ur}^2}.
			\label{eq:CRB_entries_J}
		\end{aligned}
	\end{equation}
	By observing \eqref{eq:appA_Qprime}, we need to explicitly
	evaluate the terms
	$\|\tilde{\mathbf g}\|^2$,
	$\|\tilde{\mathbf g}_u\|^2$,
	$\|\tilde{\mathbf g}_r\|^2$,
	$\tilde{\mathbf g}^{\mathrm H}\tilde{\mathbf g}_u$,
	$\tilde{\mathbf g}^{\mathrm H}\tilde{\mathbf g}_r$, and
	$\tilde{\mathbf g}_u^{\mathrm H}\tilde{\mathbf g}_r$,
	where
	$\tilde{\mathbf g}_u
	=\frac{\partial\tilde{\mathbf g}}{\partial u}$
	and
	$\tilde{\mathbf g}_r
	=\frac{\partial\tilde{\mathbf g}}{\partial r}$.
	For notational convenience, we denote
	$
	\mathbf a_{{\mathrm t},u}=\frac{\partial \mathbf a_{\mathrm t}}{\partial u}
	$, $
	\mathbf a_{{\mathrm t},r}=\frac{\partial \mathbf a_{\mathrm t}}{\partial r}
	$, $
	\mathbf b_{{\mathrm r},u}=\frac{\partial \mathbf b_{\mathrm r}}{\partial u}
	$, $
	\mathbf b_{{\mathrm r},r}=\frac{\partial \mathbf b_{\mathrm r}}{\partial r}
	$
	and
	$
	\mathbf q=\mathbf S^{\mathrm T}\mathbf a_{\mathrm t}
	$, $
	\mathbf q_u=\mathbf S^{\mathrm T}\mathbf a_{{\mathrm t},u}
	$, $
	\mathbf q_r=\mathbf S^{\mathrm T}\mathbf a_{{\mathrm t},r}$. We also denote the first derivative of $\beta(r)$ by $\beta'(r)$.
	Define
	$
	\psi_u=\frac{\partial \psi(x;u,r)}{\partial u}
	$ and $
	\psi_r=\frac{\partial \psi(x;u,r)}{\partial r}.
	$ 
	Then
	$
	[\mathbf a_{{\mathrm t},u}]_m
	=
	\jmath\frac{2\pi}{\lambda}\psi_u^{(\mathrm{t})}[\mathbf a_{\mathrm t}]_m,
	$ $
	[\mathbf a_{{\mathrm t},r}]_m
	=
	\jmath\frac{2\pi}{\lambda}\psi_r^{(\mathrm{t})}[\mathbf a_{\mathrm t}]_m.$
	Accordingly,
	\begin{equation}
		\begin{aligned}
			\|\tilde{\mathbf g}\|^2
			&=
			\beta^2(r)TPN,
			\\
			\|\tilde{\mathbf g}_u\|^2
			&=
			\beta^2(r)\frac{4\pi^2}{\lambda^2}TPN
			\Bigg[
			\left(
			s_2+\frac{2u}{r}s_3+\frac{u^2}{r^2}s_4
			\right)
			+
			m_2
			\\
			&\quad+\frac{2u}{r}m_3+
			\frac{u^2}{r^2}m_4
			+
			2\left(
			s_1+\frac{u}{r}s_2
			\right)
			\left(
			m_1+\frac{u}{r}m_2
			\right)
			\Bigg],
			\\
			\|\tilde{\mathbf g}_r\|^2
			&=
			\left[\beta'(r)\right]^2TPN
			\\
			&\quad+
			\beta^2(r)
			\frac{4\pi^2}{\lambda^2}TPN
			\frac{(1-u^2)^2}{4r^4}
			\left(
			s_4+m_4+2s_2m_2
			\right),\\
			\tilde{\mathbf g}^{\mathrm H}\tilde{\mathbf g}_u
			&=
			\jmath\beta^2(r)\frac{2\pi}{\lambda}TPN
			\left[
			\left(
			s_1+\frac{u}{r}s_2
			\right)
			+
			\left(
			m_1+\frac{u}{r}m_2
			\right)
			\right],\\
			\tilde{\mathbf g}^{\mathrm H}\tilde{\mathbf g}_r
			&=
			\beta(r)\beta'(r)TPN
			+
			\jmath\beta^2(r)\frac{2\pi}{\lambda}TPN
			\frac{1-u^2}{2r^2}
			\left(
			s_2+m_2
			\right),\\
			\tilde{\mathbf g}_u^{\mathrm H}\tilde{\mathbf g}_r
			&=
			-\jmath\beta(r)\beta'(r)
			\frac{2\pi}{\lambda}TPN
			\left[
			\left(
			s_1+\frac{u}{r}s_2
			\right)
			+
			\left(
			m_1+\frac{u}{r}m_2
			\right)
			\right]
			\\
			&\quad+
			\beta^2(r)\frac{4\pi^2}{\lambda^2}TPN
			\frac{1-u^2}{2r^2}
			\Bigg[
			s_3+\frac{u}{r}s_4
			+
			m_3+\frac{u}{r}m_4
			\\
			&\quad+
			\left(
			s_1+\frac{u}{r}s_2
			\right)m_2
			+
			s_2\left(
			m_1+\frac{u}{r}m_2
			\right)
			\Bigg].
		\end{aligned}
		\label{eq:appA_tilde_gu_gr_raw}
	\end{equation}
	Substituting \eqref{eq:appA_tilde_gu_gr_raw} into
	\eqref{eq:appA_Qprime} and using
	\eqref{eq:Mr1}-\eqref{eq:Mt3}, all the terms involving
	$\beta'(r)$ cancel. Consequently, the entries of
	${\mathbf Q}'$ are given by
	\begin{equation}
		\begin{aligned}
			{Q}'_{11}
			&=
			\beta^2(r)
			\frac{4\pi^2}{\lambda^2}TPN
			\bigl(
			M_1^{(\mathrm r)}
			+
			M_1^{(\mathrm t)}
			\bigr),
			\\
			{Q}'_{22}
			&=
			\beta^2(r)
			\frac{4\pi^2}{\lambda^2}TPN
			\bigl(
			M_2^{(\mathrm r)}
			+
			M_2^{(\mathrm t)}
			\bigr),
			\\
			{Q}'_{12}
			&={Q}'_{21}=
			\beta^2(r)
			\frac{4\pi^2}{\lambda^2}TPN
			\bigl(
			M_3^{(\mathrm r)}
			+
			M_3^{(\mathrm t)}
			\bigr).
		\end{aligned}
		\label{eq:appA_tilde_Q_entries}
	\end{equation}
	Finally, since
	$
	\mathbf J
	=
	\frac{2|\alpha|^2}{\sigma^2}
	{\mathbf Q}'
	$
	and
	\eqref{eq:CRB_entries_J}, we obtain \eqref{eq.FIM-monostatic}. This completes the derivation.
	\vspace{-2mm}
	\section{Proof of Lemma~2} \label{app.Jacobian-SPEB}
	The Cartesian position parameter $\mathbf{p} = (p_1,p_2)^\mathrm{T}$ is related to
	$\bm{\eta} = (u,r)^\mathrm{T}$ through
	$
	p_1 = r u, 
	p_2 = r\sqrt{1-u^2},
	$
	whose Jacobian matrix is
	\begin{equation}
		\begin{aligned}
			\mathbf{C}(\bm{\eta})
			= \frac{\partial \mathbf{p}}{\partial \bm{\eta}}
			=
			\begin{bmatrix}
				\displaystyle \frac{\partial p_1}{\partial u} & \displaystyle \frac{\partial p_1}{\partial r} \\[0.6ex]
				\displaystyle \frac{\partial p_2}{\partial u} & \displaystyle \frac{\partial p_2}{\partial r}
			\end{bmatrix} 
			=
			\begin{bmatrix}
				r & u\\[0.6ex]
				-\dfrac{ru}{\sqrt{1-u^2}} & \sqrt{1-u^2}
			\end{bmatrix}.
		\end{aligned}
	\end{equation}
	Then, the CRB matrix for estimating $\mathbf{p}$ is given by \cite{kay1993fssp_estimation}
	$
	\mathrm{CRB}_{\mathbf{p}}
	= \mathbf{C}(\bm{\eta})\,\mathrm{CRB}_{\bm{\eta}}\,\mathbf{C}(\bm{\eta})^\mathrm{T}.
	$
	By calculating 
	$
	\mathbf{C}(\bm{\eta})^\mathrm{T} \mathbf{C}(\bm{\eta})
	=
	\begin{bmatrix}
		\dfrac{r^2}{1-u^2} & 0\\[0.6ex]
		0 & 1
	\end{bmatrix}
	$,
	we finally obtain
	\begin{equation}
		\begin{aligned}
			\mathrm{SPEB}(\mathbf{x}_{\mathrm t},\mathbf{x}_{\mathrm r},\mathbf{p})
			&= 
			\tr\!\big(\mathrm{CRB}_{\bm{\eta}}\,\mathbf{C}(\bm{\eta})^\mathrm{T} \mathbf{C}(\bm{\eta})\big) \\
			&=
			\frac{r^2}{1-u^2}\,\mathrm{CRB}_u
			+ \mathrm{CRB}_r.
		\end{aligned}
	\end{equation}
	\section{Proof of Proposition~1}
	\label{prop.symmetry-mn}
	We observe from \eqref{eq:Mr1}--\eqref{eq:Mt3} that the quantities
	determining the entries of $\mathbf{J}$ depend only on the moments
	$m_k$ and $s_k$ of $X_{\mathrm r}$ and $X_{\mathrm t}$, respectively.
	Moreover, for each $\ell\in\{\mathrm t,\mathrm r\}$, the corresponding
	expressions can be equivalently rewritten in terms of variances and
	covariances, i.e.,
	\begin{equation}
		\begin{aligned}
			\label{eq: MR1-Mt3_var_cov}
			M_1^{(\ell)}
			&=
			\mathrm{Var}(X_\ell)
			+\frac{2u}{r}\mathrm{Cov}(X_\ell,X_\ell^2)
			+\frac{u^2}{r^2}\mathrm{Var}(X_\ell^2),\\
			M_2^{(\ell)}
			&=
			\frac{(1-u^2)^2}{4r^4}\mathrm{Var}(X_\ell^2),\\
			M_3^{(\ell)}
			&=
			\frac{1-u^2}{2r^2}
			\left(
			\mathrm{Cov}(X_\ell,X_\ell^2)+\frac{u}{r}\mathrm{Var}(X_\ell^2)
			\right).
		\end{aligned}
	\end{equation}
	Then, we introduce the covariance matrix $\mathbf \Sigma\!\big(w_\ell(\cdot)\big) \in \mathbb S_{++}^2$
	\begin{equation}
		\mathbf \Sigma\!\big(w_\ell(\cdot)\big)
		=
		\begin{bmatrix}
			\mathrm{Var}(X_\ell) & \mathrm{Cov}(X_\ell,X_\ell^2)\\
			\mathrm{Cov}(X_\ell,X_\ell^2) & \mathrm{Var}(X_\ell^2)
		\end{bmatrix}.
	\end{equation}
	which can further re-express \eqref{eq: MR1-Mt3_var_cov} in a quadratic form. Specifically,
	\begin{equation}
		\begin{aligned}
			M_1^{(\ell)}
			&=
			\mathbf{h}_u^{\mathrm T}\mathbf \Sigma\!\big(w_\ell(\cdot)\big)\mathbf{h}_u,\quad
			M_2^{(\ell)}
			=
			\mathbf{h}_r^{\mathrm T}\mathbf \Sigma\!\big(w_\ell(\cdot)\big)\mathbf{h}_r,\\
			M_3^{(\ell)}
			&=
			\mathbf{h}_u^{\mathrm T}\mathbf \Sigma\!\big(w_\ell(\cdot)\big)\mathbf{h}_r.
		\end{aligned}
	\end{equation}
	where
	$
	\mathbf{h}_u=
	\begin{bmatrix}
		1\\
		u/r
	\end{bmatrix}
	$, $
	\mathbf{h}_r=
	\begin{bmatrix}
		0\\
		(1-u^2)/(2r^2)
	\end{bmatrix}.
	$
	Therefore, the entries of $\mathbf J$ are given by
	\begin{equation}
		\begin{aligned}
			J_{uu}
			&=
			\frac{1}{\kappa(r)}
			\sum_{\ell\in\{r,t\}}
			\mathbf{h}_u^{\mathrm T}\mathbf \Sigma\!\big(w_\ell(\cdot)\big)\mathbf{h}_u,\\
			J_{rr}
			&=
			\frac{1}{\kappa(r)}
			\sum_{\ell\in\{r,t\}}
			\mathbf{h}_r^{\mathrm T}\mathbf \Sigma\!\big(w_\ell(\cdot)\big)\mathbf{h}_r,\\
			J_{ur}
			&=
			\frac{1}{\kappa(r)}
			\sum_{\ell\in\{r,t\}}
			\mathbf{h}_u^{\mathrm T}\mathbf \Sigma\!\big(w_\ell(\cdot)\big)\mathbf{h}_r.
		\end{aligned}
	\end{equation}
	Introduce the aggregate covariance matrix
	$   \mathbf \Sigma\!\big(w_{\mathrm r}(\cdot),w_{\mathrm t}(\cdot)\big)
	=
	\mathbf \Sigma\!\big(w_{\mathrm r}(\cdot)\big)
	+
	\mathbf \Sigma\!\big(w_{\mathrm t}(\cdot)\big).$
	Then, $\mathbf{J}$ can be written as
	\begin{equation}
		\mathbf{J}
		=
		\frac{1}{\kappa(r)}
		\mathbf{H}^{\mathrm T}
		\mathbf \Sigma\!\big(w_{\mathrm r}(\cdot),w_{\mathrm t}(\cdot)\big)
		\mathbf{H}.
	\end{equation}
	where
	$
	\mathbf{H}=\begin{bmatrix}
		\mathbf{h}_u & \mathbf{h}_r
	\end{bmatrix}.
	$
	For notational simplicity, denote
	$
	\mathbf \Sigma
	=
	\Sigma\!\big(w_\ell(\cdot)\big)
	$, then
	$
	\mathbf{J}^{-1}
	=
	\kappa(r)
	\mathbf{H}^{-1}
	\mathbf{\Sigma}^{-1}
	(\mathbf{H}^{-1})^{\mathrm T}.
	$
	Then, by letting \(\mathbf{e}_1=[1,0]^{\mathrm T}\), \(\mathbf{e}_2=[0,1]^{\mathrm T}\), and
	$
	\mathbf{v}(u,r)=[2ur,-2r^2]^{\mathrm T}$, we have
	\begin{equation}
		\begin{aligned}
			\mathrm{CRB}_u
			&=
			\mathbf{e}_1^{\mathrm T}\mathbf{J}^{-1}\mathbf{e}_1
			=
			\kappa(r)\mathbf{e}_1^{\mathrm T}\mathbf{\Sigma}^{-1}\mathbf{e}_1,\\
			\mathrm{CRB}_r
			&=
			\mathbf{e}_2^{\mathrm T}\mathbf{J}^{-1}\mathbf{e}_2
			=
			\frac{\kappa(r)}{(1-u^2)^2}
			\mathbf{v}(u,r)^{\mathrm T}\mathbf{\Sigma}^{-1}\mathbf{v}(u,r).
		\end{aligned}
	\end{equation}
	Accordingly, the derived SPEB can be further expressed as
	\begin{equation}
		\mathrm{SPEB}
		=
		\kappa(r)\frac{r^2}{1-u^2}\mathbf{e}_1^{\mathrm T}\mathbf{\Sigma}^{-1}\mathbf{e}_1
		+
		\frac{\kappa(r)}{(1-u^2)^2}
		\mathbf{v}(u,r)^{\mathrm T}\mathbf{\Sigma}^{-1}\mathbf{v}(u,r).
	\end{equation}
	This motivates the linear functional: for any
	\(\mathbf{Z}\in\mathbb S_{++}^2\),
	\begin{equation}
		\begin{aligned} 
			\label{eq.functional_linear}
			L_{u,r}(\mathbf{Z})
			=
			\kappa(r)\frac{r^2}{1-u^2}\mathbf{e}_1^{\mathrm T}\mathbf{Z}\mathbf{e}_1
			+
			\frac{\kappa(r)}{(1-u^2)^2}
			\mathbf{v}(u,r)^{\mathrm T}\mathbf{Z}\mathbf{v}(u,r).
		\end{aligned}
	\end{equation}
	The worst-case objective can be accordingly written as
	$	F\!\big(w_{\mathrm r}(\cdot),w_{\mathrm t}(\cdot)\big)
	=
	\max_{(u,r)\in\Gamma}
	L_{u,r}
	\left(
	\mathbf{\Sigma}^{-1}
	\right).$
	This representation collects all distribution-dependent moment information into \(\boldsymbol{\Sigma}\), and \(L_{u,r}(\cdot)\) maps \(\boldsymbol{\Sigma}^{-1}\) to the corresponding SPEB. Therefore, the impact of symmetrization can be analyzed through its induced change in \(\boldsymbol{\Sigma}\). Accordingly, 
	for each $\ell\in\{\mathrm t,\mathrm r\}$, define its reflected distribution
	$
	\check w_\ell(\zeta)=w_\ell(-\zeta),
	$
	and its symmetrized distribution
	$
	w_{\ell,\mathrm{sym}}(\zeta)
	=
	\frac12\big(w_\ell(\zeta)+\check w_\ell(\zeta)\big).
	$
	Since $w_{\ell,\mathrm{sym}}(\cdot)$ is the average of $w_\ell(\cdot)$ and its reflection,
	we have
	$
		\mathbf \Sigma\!\big(w_{\ell,\mathrm{sym}}(\cdot)\big)
		=
		\begin{bmatrix}
			\mathbb E[X_\ell^2] & 0\\
			0 & \mathrm{Var}(X_\ell^2)
		\end{bmatrix}.
	$
	On the other hand, reflection preserves $\mathrm{Var}(X_\ell)$ and
	$\mathrm{Var}(X_\ell^2)$, while flipping the sign of
	$\mathrm{Cov}(X_\ell,X_\ell^2)$. Therefore,
	\begin{equation}
		\frac12\Big(
		\mathbf \Sigma\!\big(w_\ell(\cdot)\big)
		+
		\mathbf \Sigma\!\big(\check w_\ell(\cdot)\big)
		\Big)
		=
		\begin{bmatrix}
			\mathrm{Var}(X_\ell) & 0\\
			0 & \mathrm{Var}(X_\ell^2)
		\end{bmatrix}.
	\end{equation}
	It follows that
	\begin{equation}
		\label{eq: 47}
		\mathbf \Sigma\!\big(w_{\ell,\mathrm{sym}}(\cdot)\big)=
		\frac12\Big(
		\mathbf \Sigma\!\big(w_\ell(\cdot)\big)
		+
		\mathbf \Sigma\!\big(\check w_\ell(\cdot)\big)
		\Big)
		+
		\begin{bmatrix}
			\mathbb E^2[X_\ell]& 0\\
			0 & 0
		\end{bmatrix}.
	\end{equation} 	
	For notational simplicity, we denote
	$\check{\mathbf \Sigma}
	=
	\mathbf \Sigma\!\big(\check w_{\mathrm r}(\cdot),\check w_{\mathrm t}(\cdot)\big)
	$ and $
	\mathbf \Sigma_{\mathrm{sym}}
	=
	\mathbf \Sigma\!\big(w_{\mathrm r,\mathrm{sym}}(\cdot),w_{\mathrm t,\mathrm{sym}}(\cdot)\big).
	$
	Summing the \eqref{eq: 47} over $\ell\in\{\mathrm r, \mathrm{t}\}$ yields
	$
	\mathbf \Sigma_{\mathrm{sym}}
	\succeq
	\frac12\big(\mathbf \Sigma+\check{\mathbf \Sigma}\big).
	$
	Since the matrix inverse is monotone decreasing on $\mathbb S_{++}^2$, hence
	$\mathbf \Sigma_{\mathrm{sym}}^{-1}
	\preceq
	\left(\frac{\mathbf \Sigma+\check{\mathbf \Sigma}}{2}\right)^{-1}.$
	Moreover, the matrix inverse is operator convex on $\mathbb S_{++}^2$, and hence
	$
	\left(\frac{\mathbf \Sigma+\check{\mathbf \Sigma}}{2}\right)^{-1}
	\preceq
	\frac12\big(\mathbf \Sigma^{-1}+\check{\mathbf \Sigma}^{-1}\big),
	$
	yielding
	$
	\mathbf \Sigma_{\mathrm{sym}}^{-1}
	\preceq
	\frac12\big(\mathbf \Sigma^{-1}+\check{\mathbf \Sigma}^{-1}\big).
	$
	
	We next note that the mapping $\mathbf{Z}\mapsto L_{u,r}(\mathbf{Z})$ is linear and
	Loewner-monotone on $\mathbb S_{++}^2$, since both coefficients
	$\kappa(r)\frac{r^2}{1-u^2}$ and $\frac{\kappa(r)}{(1-u^2)^2}$ are nonnegative.
	Therefore,
	\begin{equation}
		\begin{aligned}
			L_{u,r}\!\big(\mathbf \Sigma_{\mathrm{sym}}^{-1}\big)
			\le
			\frac12\Big(
			L_{u,r}(\mathbf \Sigma^{-1})
			+
			L_{u,r}(\check{\mathbf \Sigma}^{-1})
			\Big).
		\end{aligned}
	\end{equation}
	Taking the maximum over $(u,r)\in\Gamma$ gives
	\begin{equation}
		\begin{aligned}
			\max_{(u,r)\in\Gamma}
			L_{u,r}\!\big(\mathbf \Sigma_{\mathrm{sym}}^{-1}\big)
			\le
			\frac12
			\max_{(u,r)\in\Gamma}
			\Big(
			L_{u,r}(\mathbf \Sigma^{-1})
			+
			L_{u,r}(\check{\mathbf \Sigma}^{-1})
			\Big).
		\end{aligned}
	\end{equation}
	Thus $F\!\big(w_{\mathrm r,\mathrm{sym}}(\cdot),w_{\mathrm t,\mathrm{sym}}(\cdot)\big)\le\frac12\Big(
	F\!\big(w_{\mathrm r}(\cdot),w_{\mathrm t}(\cdot)\big)
	+
	F\!\big(\check w_{\mathrm r}(\cdot),\check w_{\mathrm t}(\cdot)\big)$. To this end, let
	$\mathbf{D}=\mathrm{diag}(-1,1).$
	From the reflection property discussed above, we have
	$	\check{\mathbf \Sigma}=\mathbf{D}\mathbf \Sigma \mathbf{D}$,
	and therefore
	$
	\check{\mathbf \Sigma}^{-1}=\mathbf{D}\mathbf \Sigma^{-1}\mathbf{D}.
	$
	On the other hand, since $\mathbf{D}\mathbf{e}_1=-\mathbf{e}_1$ and
	$
	\mathbf{D}\mathbf{v}(u,r)=\mathbf{v}(-u,r),
	$
	we have
	$
	L_{u,r}(\mathbf{D}\mathbf{Z}\mathbf{D})=L_{-u,r}(\mathbf{Z}).
	$
	Thus,
	\begin{equation}
			F\!\big(\check w_{\mathrm r}(\cdot),\check w_{\mathrm t}(\cdot)\big)=
			\max_{(u,r)\in\Gamma}
			L_{u,r}(\mathbf{D}\mathbf \Sigma^{-1}\mathbf{D}) =
			\max_{(u,r)\in\Gamma}
			L_{-u,r}(\mathbf \Sigma^{-1}).
	\end{equation}
	Since $\Gamma$ is symmetric with respect to the transformation
	$(u,r)\mapsto(-u,r)$, the maximization set remains invariant, and hence
	$
	F\!\big(\check w_{\mathrm r}(\cdot),\check w_{\mathrm t}(\cdot)\big)
	=
	F\!\big(w_{\mathrm r}(\cdot),w_{\mathrm t}(\cdot)\big).
	$
	Substituting this into the previous inequality yields
	$
	F\!\big(w_{\mathrm r,\mathrm{sym}}(\cdot),w_{\mathrm t,\mathrm{sym}}(\cdot)\big)
	\le
	F\!\big(w_{\mathrm r}(\cdot),w_{\mathrm t}(\cdot)\big),
	$
	which proves Proposition~1.
	
	\section{Proof of Proposition~2}
	\label{prop.worstpoint-mn}
	Define $
	T(u,r)
	=
	\frac{k_1}{\Xi_1}
	+
	\frac{k_2}{\Xi_2},
	$
	such that the objective function in
	\eqref{eq:mono-SPEB-sym} is
	$\kappa(r)T(u,r)$.
	Since $\beta(r)$ is non-increasing in $r$,
	$\kappa(r)$ is non-decreasing in $r$.
	Moreover, $T(u,r)$ is strictly increasing in $r$ for every
	fixed $u$. Hence, for any $r_2>r_1$,
	\begin{equation}
		\begin{aligned}
			\kappa(r_2)T(u,r_2)
			\ge
			\kappa(r_1)T(u,r_2)
			>
			\kappa(r_1)T(u,r_1).
		\end{aligned}
	\end{equation}
	Therefore, the maximum over $r$ for every fixed $u$ is attained
	at the Rayleigh boundary
	$r=d_{\mathrm R}(u)
	=d_{\mathrm R,\max}(1-u^2)$.
	Substituting $r=d_{\mathrm{R},\max}(1-u^2)$, the coefficients $k_1,k_2$ and variables $\Xi_1, \Xi_2$ from \eqref{eq:mono-SPEB-sym} into
	$T(u,r)$ yields 
	\begin{equation}
		\begin{aligned}
			\label{App. eq. T(u)}
			T(u,d_\mathrm{R})
			=
			\frac{d_{\mathrm{R},\max}^2(1+3u^2)}{\Xi_1}
			+
			\frac{4d_{\mathrm {R},\max}^4(1-u^2)^2}{\Xi_2}.
		\end{aligned}
	\end{equation}
	Since $T(u,d_{\mathrm R})$ depends on $u$ only through
	$u^2\in[0,1]$ and
	$
	\frac{\mathrm d^2 T(u,d_{\mathrm R})}
	{\mathrm d (u^2)^2}
	=
	\frac{8d_{\mathrm R,\max}^4}{\Xi_2}
	>0,
	$
	it is strictly convex with respect to $u^2$. Therefore, its maximum
	over $u\in[-1,1]$ is attained at $u^2\in\{0,1\}$, i.e.,
	$u_{\mathrm{worst}}^2\in\{0,1\}$.
	It remains to compare the two endpoint values. Define
	$	\Delta_{01}
	=
	T\bigl(0, d_\mathrm{R}\bigr)
	-
	T\bigl(1, d_\mathrm{R}\bigr).$
	Then
	$
	\Delta_{01}
	=
	4d_{\mathrm {R},\max}^2
	\left(
	\frac{d_{\mathrm{R},\max}^2}{\Xi_2}
	-
	\frac{3}{4\Xi_1}
	\right).
	$
	
	Since both distributions are supported on $[-a,a]$, for each 
	$\ell\in\{\mathrm t, \mathrm r\}$, we have $	\mathrm{Var}(X_\ell^2)	\le	\mathbb{E}(X_\ell^4)	\le a^2\mathbb{E}(X_\ell^2).$
	Moreover, under centro-symmetry, $\mathbb{E}(X_\ell)=0$, and hence
	$\mathbb{E}(X_\ell^2)=\mathrm{Var}(X_\ell)$.
	Therefore,
	\begin{equation}
		\begin{aligned}
			\Xi_2
			&=
			\mathrm{Var}(X_{\mathrm r}^2)+\mathrm{Var}(X_{\mathrm t}^2)\\
			&\le
			a^2\mathrm{Var}(X_{\mathrm r})+a^2\mathrm{Var}(X_{\mathrm t})
			=
			a^2\Xi_1.
		\end{aligned}
	\end{equation}
	It follows that
	$
	\Delta_{01}
	\ge
	\frac{d_{\mathrm{R},\max}^2}{\Xi_1}
	\left(
	\frac{4d_{\mathrm {R},\max}^2}{a^2}-3
	\right).
	$
	Since practical MIMO arrays span multiple wavelengths, they
	naturally satisfy $a>\frac{\sqrt{3}}{16}\lambda$, meaning $d_{\mathrm{R},\max}>\frac{\sqrt{3}}{2}a$.
	Hence $\Delta_{01}\bigl(w_{\mathrm r}(\cdot),w_{\mathrm t}(\cdot)\bigr) >0$, i.e.,
	$
	T(0, d_\mathrm{R})
	>
	T(1, d_\mathrm{R}).
	$
	Moreover, since
	$d_{\mathrm R}(u)=d_{\mathrm R,\max}(1-u^2)
	\le d_{\mathrm R,\max}$
	and $\kappa(r)$ is non-decreasing in $r$, we have
	$
	\kappa\bigl(d_{\mathrm R}(u)\bigr)
	\le
	\kappa(d_{\mathrm R,\max}).
	$
	Consequently, for every $u\neq0$,
	$
	\kappa\bigl(d_{\mathrm R}(u)\bigr)
	T\bigl(u,d_{\mathrm R}(u)\bigr)
	<
	\kappa(d_{\mathrm R,\max})
	T(0,d_{\mathrm R,\max}).
	$
	Therefore, the worst-case target location is uniquely given by
	$
	\bm{\eta}_{\mathrm{worst}}
	=
	[\,0,\ d_{\mathrm R,\max}\,]^{\mathrm T},$
	which proves Proposition~2.
	\label{app.proposition2}
	\section{Proof of Lemma~3}
	\label{lemma:three-point}
	Before proving the lemma, we note that for any fixed feasible distribution on one side, replacing
	the distribution on the other side by another distribution that
	preserves all the relevant moments leaves the objective value unchanged.
	We first establish such a moment-preserving replacement for the
	receive-side distribution by invoking the Richter-Tchakaloff Theorem
	\cite{schmüdgen2017}.
	It says: 
	
	Suppose that $(\mathcal{Y},\mu)$ is a measure space, and let 
	$V \subset L^1_{\mathbb{R}}(\mathcal{Y},\mu)$ be a finite-dimensional 
	linear subspace with $\dim(V)$. 
	Let $L^\mu:V\to\mathbb{R}$ be the moment functional defined by
	\begin{equation}
		\begin{aligned}
			L^\mu(f) &= \int f \, d\mu, \qquad f\in V.
		\end{aligned}
	\end{equation}
	Then there exists a $k$-atomic measure
	\begin{equation}
		\begin{aligned}
			\nu &= \sum_{j=1}^k z_j \delta_{x_j} \in M_+(\mathcal{Y}), 
			\qquad k\le \dim(V),
		\end{aligned}
	\end{equation}
	such that $L^\mu = L^\nu$, i.e.,
	\begin{equation}
		\begin{aligned}
			\int f \, d\mu 
			&= \int f \, d\nu 
			\equiv \sum_{j=1}^k z_j f(x_j), 
			\qquad \forall f\in V.
		\end{aligned}
	\end{equation}
	where $x_j\in\mathcal{Y}$ and $z_j\ge0$ denote the support points and
	their associated weights, respectively. This theorem implies that, if an objective function and its constraints
	depend on a measure only through the moments associated with a
	finite-dimensional function space \(V\), then any feasible moment vector
	generated by a general positive measure can be exactly reproduced by a
	positive atomic measure supported on at most \(\dim(V)\) points. 
	Now we apply this theorem to prove Lemma~3.
	
	Observing \eqref{eq:opt-sym-broadside-mn}, on the receiver side, it depends on the measure $\mu$ solely through the moments $(m_0, m_2, m_4)$. Let $V = \operatorname{span}\{1, x^2, x^4\} \subset L^1_{\mathbb{R}}([-a,a])$ be the subspace with $\dim(V) = 3$. According to the Richter-Tchakaloff Theorem, for any arbitrary probability measure $\mu$ on $[-a,a]$, there exists a $k$-atomic measure $\nu$ with $k \le \dim(V) = 3$ such that $\nu$ preserves the moments of $\mu$, the objective values for both measures are identical:
	\begin{equation}
		\begin{aligned}
			\varphi \big(m_0(\nu),m_2(\nu),m_4(\nu)\big)
			= \varphi \big(m_0(\mu),m_2(\mu),m_4(\mu)\big),
		\end{aligned}
	\end{equation}
	where $\varphi(\cdot)$ represents the mapping from the moment vector $(m_0, m_2, m_4)$ to the objective value.
	Similarly, the transmit side depends on the measure solely through the moments $(s_0, s_2, s_4)$ for any $w_{\mathrm r}(\cdot)$, thereby the optimal transmit-side distribution can also be restricted to utmost three atoms. Consequently, for both the distributions $w_{\mathrm t}(\cdot)$ and $w_{\mathrm r}(\cdot)$, the performance achievable by a general probability measure can be equally attained by a discrete measure supported on at most three points. Therefore we can restrict the optimization domain to the set of measures with at most three atoms. This completes the proof.
	\label{app.lemma3}
	\section{Proof of Lemma~4}
	\label{app.lemma4}
	Consider an arbitrary feasible pair of centro-symmetric distributions
	$\big(w_{\mathrm r}(\cdot),w_{\mathrm t}(\cdot)\big)$ in problem
	(P3). We first fix $w_{\mathrm t}(\cdot)$ and show that
	$w_{\mathrm r}(\cdot)$ can be replaced by a centro-symmetric
	distribution supported on $\{-a,0,+a\}$ without increasing the
	objective value.
	Let $\mu$ be the probability measure associated with
	$w_{\mathrm r}(\cdot)$. Since $x^4\leq a^2x^2$ for all
	$x\in[-a,a]$, we have
	$
	m_4\leq a^2m_2.
	$
	By the Cauchy-Schwarz inequality,
	\begin{equation}
		m_2^2
		= \Bigl(\int x^2 d\mu(x)\Bigr)^2
		\leq \Bigl(\int x^4 d\mu(x)\Bigr)
		\Bigl(\int 1\,d\mu(x)\Bigr)
		= m_4.
	\end{equation}
	Let $\mathcal{J}\bigl(w_{\mathrm r},w_{\mathrm t}\bigr)$ denote the
	objective function in \eqref{eq:mono-SPEB-sym}. For any fixed feasible
	transmit-side distribution $w_{\mathrm t}(\cdot)$ and fixed $m_2$,
	$\mathcal{J}\bigl(w_{\mathrm r},w_{\mathrm t}\bigr)$ is strictly decreasing over
	$m_4 \in [m_2^2,a^2m_2]$.
	The largest feasible value
	$m_4=a^2m_2$ is attained by the centro-symmetric probability distribution $\tilde w_{\mathrm r}$ supported on $\{-a,0,+a\}$ with probability masses
	$\{\frac{m_2}{2a^2},\left(1-\frac{m_2}{a^2}\right),	\frac{m_2}{2a^2}\}$,
	whose second-
	and fourth-order moments satisfy
	$
	\tilde m_2=m_2,
	$ and $
	\tilde m_4=a^2m_2\geq m_4.
	$
	Consequently,
	$
	\tilde m_2+s_2=m_2+s_2
	$
	and
	$
	(\tilde m_4-\tilde m_2^2)+(s_4-s_2^2)
	\geq
	(m_4-m_2^2)+(s_4-s_2^2).
	$
	Therefore
	$
	\mathcal{J}\bigl(\tilde w_{\mathrm r},w_{\mathrm t}\bigr)
	\leq
	\mathcal{J}\bigl(w_{\mathrm r},w_{\mathrm t}\bigr).
	$
	The same construction applies to the transmit-side distribution for
	any fixed feasible receive-side distribution. Therefore, there exists
	a centro-symmetric distribution $\tilde w_{\mathrm t}(\cdot)$
	supported on $\{-a,0,+a\}$ such that
	$
			\mathcal{J}\bigl(
			\tilde w_{\mathrm r},\tilde w_{\mathrm t}
			\bigr)
			\leq
			\mathcal{J}\bigl(
			\tilde w_{\mathrm r},w_{\mathrm t}
			\bigr)
			\leq
			\mathcal{J}\bigl(
			w_{\mathrm r},w_{\mathrm t}
			\bigr).
	$
	Consequently, every feasible distribution pair admits a pair of
	centro-symmetric distributions with support contained in
	$\{-a,0,+a\}$ and with no larger objective value. Hence, the minimum
	of problem (P3) can be attained within $\{-a,0,+a\}$. This
	completes the proof.
	\section{Proof of Lemma~\ref{lem:qtqr_active}}
	\label{app.lem:qtqr_active}
	For any feasible pair \((q_{\mathrm r},q_{\mathrm t})\), define $\bar q=\frac{q_{\mathrm r}+q_{\mathrm t}}{2}$.
	The term $\frac{d_{\mathrm{R},\max}^2}
	{a^2(q_{\mathrm r}+q_{\mathrm t})}$ of 
	\eqref{eq:active_J_qt_qr}, which depends only on $q_{\mathrm r}+q_{\mathrm t}$, remains unchanged when 
	$(q_{\mathrm r},q_{\mathrm t})$ is replaced by $(\bar q,\bar q)$. It remains to compare the first denominator. We have
	\begin{equation}
		\begin{aligned}
			&q_{\mathrm r}(1-q_{\mathrm r})+q_{\mathrm t}(1-q_{\mathrm t})
			=2\bar q(1-\bar q)-\frac{(q_{\mathrm r}-q_{\mathrm t})^2}{2}
			\leq 2\bar q(1-\bar q),
		\end{aligned}
	\end{equation}
	with equality if and only if \(q_{\mathrm r}=q_{\mathrm t}\). Let $H(q_{\mathrm r},q_{\mathrm t})$ denote the objective function
	in \eqref{eq:active_J_qt_qr}. Since the coefficient of the reciprocal 
	of this denominator in \(H(q_{\mathrm r},q_{\mathrm t})\) is positive, it follows that $H(q_{\mathrm r},q_{\mathrm t})\geq H(\bar q,\bar q),$ with strict inequality whenever \(q_{\mathrm r}\neq q_{\mathrm t}\). Therefore, every global minimizer of \eqref{eq:active_J_qt_qr} must satisfy $q_{\mathrm r}^\star=q_{\mathrm t}^\star$.
	This completes the proof.
	\bibliographystyle{IEEEtran}
	\bibliography{references}   

@book{kay1993fssp_estimation,
	author    = {Kay, Steven M.},
	title     = {Fundamentals of Statistical Signal Processing: Estimation Theory},
	publisher = {PTR Prentice-Hall},
	year      = {1993},
	isbn      = {0133457117}
}

@article{ChenJCIN2025,
	author       = {Lin Chen and
	Chang Cai and
	Huiyuan Yang and
	Xiaojun Yuan and
	Ying{-}Jun Angela Zhang},
	title        = {Scoring {ISAC:} Benchmarking Integrated Sensing and Communications
	via Score-Based Generative Modeling},
	journal      = {J. Commun. Inf. Networks},
	volume       = {10},
	number       = {3},
	pages        = {224-245},
	month = {Sept.},
	year         = {2025}
}

@book{Balanis,
	author    = {C. A. Balanis},
	title     = {Antenna Theory Analysis and Design},
	publisher = {Wiley-Interscience},
	year      = {2005},
	isbn      = {978-1-118-64206-1}
}

@misc{Song2026arxiv,
	title={Integrated Sensing and Communications for Low-Altitude Economy with Deterministic Sensing and {Gaussian} Information Signals}, 
	author={Xianxin Song and Xianghao Yu and Jie Xu and Derrick Wing Kwan Ng},
	year={2026},
	eprint={2604.19040},
	archivePrefix={arXiv},
	primaryClass={eess.SP},
	url={https://arxiv.org/abs/2604.19040}, 
}

@ARTICLE{MaICST2026,
		author={Ma, Wenyan and Zhu, Lipeng and Tan, Yanhua and Zheng, Beixiong and Zhang, Yujie and Zhang, Yuchen and Ying, Keke and Gao, Zhen and Sun, He and Shao, Xiaodan and Xiao, Zhenyu and Niyato, Dusit and Zhang, Rui},
		journal={IEEE Commun. Surveys Tuts.}, 
		title={A Survey on Reconfigurable and Movable Antennas for Wireless Communications and Sensing}, 
		year={2026},
		volume={28},
		number={},
		pages={4842-4882}}

@ARTICLE{YuanOJAP2023,
	author={Yuan, Shuai S. A. and Chen, Xiaoming and Huang, Chongwen and Sha, Wei E. I.},
	journal={IEEE Open J. Antennas Propag.}, 
	title={Effects of Mutual Coupling on Degree of Freedom and Antenna Efficiency in Holographic {MIMO} Communications}, 
	year={2023},
	volume={4},
	number={},
	pages={237-244}}

@ARTICLE{NewCST2025,
		author={New, Wee Kiat and Wong, Kai-Kit and Xu, Hao and Wang, Chao and Ghadi, Farshad Rostami and Zhang, Jichen and Rao, Junhui and Murch, Ross and Ramírez-Espinosa, Pablo and Morales-Jimenez, David and Chae, Chan-Byoung and Tong, Kin-Fai},
		journal={IEEE Commun. Surveys Tuts.}, 
		title={A Tutorial on Fluid Antenna System for {6G} Networks: Encompassing Communication Theory, Optimization Methods and Hardware Designs}, 
		month = {Aug.},
		year={2025},
		volume={27},
		number={4},
		pages={2325-2377}}

@ARTICLE{DuCST2025,
	author={Du, Rui and Hua, Haocheng and Xie, Hailiang and Song, Xianxin and Lyu, Zhonghao and Hu, Mengshi and Narengerile and Xin, Yan and McCann, Stephen and Montemurro, Michael and Han, Tony Xiao and Xu, Jie},
	journal={IEEE Commun. Surveys Tuts.}, 
	title={An Overview on {IEEE} 802.11bf: {WLAN} Sensing}, 
	month = {Feb.},
	year={2025},
	volume={27},
	number={1},
	pages={184-217}}

@ARTICLE{NionTSP2010,
		author={Nion, Dimitri and Sidiropoulos, Nicholas D.},
		journal={IEEE Trans. Signal Process.}, 
		title={Tensor Algebra and Multidimensional Harmonic Retrieval in Signal Processing for {MIMO} Radar}, 
		month = {Nov.},
		year={2010},
		volume={58},
		number={11},
		pages={5693-5705},
		doi={10.1109/TSP.2010.2058802}}

@ARTICLE{JornetProIEEE2025,
	author={Jornet, Josep M. and Petrov, Vitaly and Wang, Hua and Popović, Zoya and Shakya, Dipankar and Siles, Jose V. and Rappaport, Theodore S.},
	journal={Proc. IEEE}, 
	title={The Evolution of Applications, Hardware Design, and Channel Modeling for Terahertz ({THz}) Band Communications and Sensing: Ready for 6{G}?}, 
	month = {Sept.},
	year={2025},
	volume={113},
	number={9},
	pages={920-951}}

@ARTICLE{ZhuCST,
	author={Zhu, Lipeng and Ma, Wenyan and Mei, Weidong and Zeng, Yong and Wu, Qingqing and Ning, Boyu and Xiao, Zhenyu and Shao, Xiaodan and Zhang, Jun and Zhang, Rui},
	journal={IEEE Commun. Surveys Tuts.}, 
	title={A Tutorial on Movable Antennas for Wireless Networks}, 
	month = {4th Quart.},
	year={2025},
	volume={28},
	number={},
	pages={3002-3054}}

@ARTICLE{ZhuCM,
	author={Zhu, Lipeng and Ma, Wenyan and Zhang, Rui},
	journal={IEEE Commun. Mag.}, 
	title={Movable Antennas for Wireless Communication: Opportunities and Challenges}, 
	month = {June},
	year={2024},
	volume={62},
	number={6},
	pages={114-120}}

@misc{wang202512arxiv,
	title={Movable Antenna Empowered Near-Field Sensing via Antenna Position Optimization}, 
	author={Yushen Wang and Weidong Mei and Xin Wei and Ya Fei Wu and Zhi Chen and Boyu Ning},
	year={2025},
	eprint={2512.00758},
	archivePrefix={arXiv},
	primaryClass={cs.IT},
	url={https://arxiv.org/abs/2512.00758}, 
	}

@ARTICLE{DingTWC2026,
	author={Ding, Jingze and Zhou, Zijian and Shao, Xiaodan and Jiao, Bingli and Zhang, Rui},
	journal={IEEE Trans. Wireless Commun.}, 
	title={Movable Antenna-Aided Near-Field Integrated Sensing and Communication}, 
	year={2026},
	volume={25},
	number={},
	pages={493-508}}

@ARTICLE{SunIoTJ2025,
	author={Sun, Yunan and Xu, Hao and Ouyang, Chongjun and Yang, Hongwen},
	journal={IEEE Internet Things J.}, 
	title={Rotatable and Movable Antenna-Enabled Near-Field Integrated Sensing and Communication}, 
	month={Nov.},
	year={2025},
	volume={12},
	number={21},
	pages={45119-45132}}

@ARTICLE{GazzahTAP2014,
	author={Gazzah, Houcem and Delmas, Jean Pierre},
	journal={IEEE Trans. Antennas Propag.}, 
	title={{CRB}-Based Design of Linear Antenna Arrays for Near-Field Source Localization}, 
	month = {Apr.},
	year={2014},
	volume={62},
	number={4},
	pages={1965-1974}}

@INPROCEEDINGS{liumeditcom2026,
	author={Liu, Jinjian and Song, Xianxin and Yu, Xianghao},
	booktitle={Proc. IEEE Int. Mediterranean Conf. Commun. Netw. (MeditCom)}, 
	title={Optimal Movable Antenna Placement for Near-Field Wireless Sensing}, 
	address = {Cagliari, Italy},
	year={2026},
	volume={},
	number={},
	pages={}}

@ARTICLE{MaTWC0824,
	author={Ma, Wenyan and Zhu, Lipeng and Zhang, Rui},
	journal={IEEE Trans. Wireless Commun.}, 
	title={Movable Antenna Enhanced Wireless Sensing via Antenna Position Optimization}, 
	month={Nov.},
	year={2024},
	volume={23},
	number={11},
	pages={16575-16589}}

@ARTICLE{liuOJCS2023,
	author={Liu, Yuanwei and Wang, Zhaolin and Xu, Jiaqi and Ouyang, Chongjun and Mu, Xidong and Schober, Robert},
	journal={IEEE Open J. Commun. Soc.}, 
	title={Near-Field Communications: A Tutorial Review}, 
	month = {Aug.},
	year={2023},
	volume={4},
	number={},
	pages={1999-2049}}

@ARTICLE{LuTWC2022,
	author={Lu, Haiquan and Zeng, Yong},
	journal={IEEE Trans. Wireless Commun.}, 
	title={Communicating With Extremely Large-Scale Array/Surface: Unified Modeling and Performance Analysis}, 
	month = {June},
	year={2022},
	volume={21},
	number={6},
	pages={4039-4053}}

@ARTICLE{JoJSAC2023,
	author={Jo, Joo-Hyun and Shim, Jae-Nam and Kim, Byoungnam and Chae, Chan-Byoung and Kim, Dong Ku},
	journal={IEEE J. Sel. Areas Commun.}, 
	title={{AoA}-Based Position and Orientation Estimation Using Lens {MIMO} in Cooperative Vehicle-to-Vehicle Systems}, 
	month = {Oct.},
	year={2023},
	volume={41},
	number={12},
	pages={3719-3735}}

@ARTICLE{ZhangTWC2022,
	author={Zhang, Haiyang and Shlezinger, Nir and Guidi, Francesco and Dardari, Davide and Imani, Mohammadreza F. and Eldar, Yonina C.},
	journal={IEEE Trans. Wireless Commun.}, 
	title={Beam Focusing for Near-Field Multiuser {MIMO} Communications}, 
	month={Sept.},
	year={2022},
	volume={21},
	number={9},
	pages={7476-7490}}

@ARTICLE{ShiTSP2021,
	author={Shi, Wanlu and Vorobyov, Sergiy A. and Li, Yingsong},
	journal={IEEE Trans. Signal Process.}, 
	title={{ULA} Fitting for Sparse Array Design}, 
	month = {Nov.},
	year={2021},
	volume={69},
	number={},
	pages={6431-6447}}

@ARTICLE{GodrichTSP0711,
		author={Godrich, Hana and Petropulu, Athina P. and Poor, H. Vincent},
		journal={IEEE Trans. Signal Process.}, 
		title={Power Allocation Strategies for Target Localization in Distributed Multiple-Radar Architectures}, 
		month={July},
		year={2011},
		volume={59},
		number={7},
		pages={3226-3240}}

@ARTICLE{WangTSP0124,
		author={Wang, Huizhi and Xiao, Zhiqiang and Zeng, Yong},
		journal={IEEE Trans. Signal Process.}, 
		title={Cramér-{Rao} Bounds for Near-Field Sensing With Extremely Large-Scale {MIMO}}, 
		month = {Jan.},
		year={2024},
		volume={72},
		number={},
		pages={701-717}}

@ARTICLE{BoyerTSP,
		author={Boyer, Rémy},
		journal={IEEE Trans. Signal Process.}, 
		title={Performance Bounds and Angular Resolution Limit for the Moving Colocated {MIMO }Radar}, 
		month = {Apr.},
		year={2011},
		volume={59},
		number={4},
		pages={1539-1552}}

@ARTICLE{BekkermanTSP,
		author={Bekkerman, I. and Tabrikian, J.},
		journal={IEEE Trans. Signal Process.}, 
		title={Target Detection and Localization Using {MIMO} Radars and Sonars}, 
		month = {Oct.},
		year={2006},
		volume={54},
		number={10},
		pages={3873-3883}}

@ARTICLE{LiTSP2008,
	author={Li, Jian and Xu, Luzhou and Stoica, Petre and Forsythe, Keith W. and Bliss, Daniel W.},
	journal={IEEE Trans. Signal Process.}, 
	title={Range Compression and Waveform Optimization for {MIMO} Radar: A {Cramér-Rao} Bound Based Study}, 
	month={Jan.},
	year={2008},
	volume={56},
	number={1},
	pages={218-232}}

@ARTICLE{YangTCCN1025,
		author={Yang, Yinchao and Zhou, Jingxuan and Yang, Zhaohui and Shikh-Bahaei, Mohammad R.},
		journal={IEEE Trans. Cogn. Commun. Netw.}, 
		title={Fluid Antenna-Enabled Near-Field Integrated Sensing, Computing, and Semantic Communication for Emerging Applications}, 
		month = {Oct.},
		year={2025},
		volume={11},
		number={5},
		pages={3062-3078}}

@ARTICLE{WerfTSP2026,
		author={van der Werf, Ids and Rajamäki, Robin and Leus, Geert},
		journal={IEEE Trans. Signal Process.}, 
		title={{CRB}-Optimal Arrays and Waveforms in Active Sensing: Role of Redundancy and Spatial Covariance of Array Geometry}, 
		month = {June},
    	year={2026},
        volume={74},
        number={},
        pages={2691-2705}}

@INPROCEEDINGS{ZhouGlobalcom2024,
		author={Zhou, Jingxuan and Yang, Yinchao and Yang, Zhaohui and Xu, Wei and Shikh-Bahaei, Mohammad},
		booktitle={Proc. IEEE Global Commun. Conf. Wkshps. (GLOBECOM Wkshps)}, 
		title={Location Optimization for Fluid Antenna-Assisted Near-Field System}, 
		address = {Cape Town, South Africa},
		year={2024},
		volume={},
		number={},
		pages={1-6}}

@ARTICLE{ZhengWC2026,
		author={Yang, Zheng and Wang, Ning and Sun, Yanshi and Ding, Zhiguo and Schober, Robert and Karagiannidis, George K. and Wong, Vincent W.S. and Dobre, Octavia A.},
		journal={IEEE Wireless Commun.}, 
		title={Pinching Antennas: Principles, Applications and Challenges},
		month = {Apr.}, 
		year={2026},
		volume={33},
		number={2},
		pages={175-184}}

@ARTICLE{Wintit2010,
	author={Shen, Yuan and Win, Moe Z.},
	journal={IEEE Trans. Inf. Theory}, 
	title={Fundamental Limits of Wideband Localization— {P}art {I}: A General Framework}, 
	month = {Sept.},
	year={2010},
	volume={56},
	number={10},
	pages={4956-4980}}

@book{schmüdgen2017,
	title={The Moment Problem},
	author={Schm{\"u}dgen, K.},
	isbn={9783319645469},
	year={2017},
	publisher={Springer International Publishing}
	}
	
\end{document}